\documentclass[11pt,a4paper]{article}

\usepackage[margin=1in]{geometry}
\usepackage[utf8]{inputenc}
\usepackage[T1]{fontenc}
\usepackage{lmodern}
\usepackage[tbtags]{mathtools}
\allowdisplaybreaks
\usepackage{amssymb,amsthm}
\usepackage{cases}
\usepackage{hyperref}
\usepackage{url}
\usepackage[svgnames]{xcolor}
\usepackage[capitalise,nameinlink]{cleveref}
\hypersetup{colorlinks={true},linkcolor={DarkBlue},citecolor=[named]{DarkGreen}}
\usepackage[font=footnotesize]{caption}
\usepackage[numbers,sort&compress]{natbib}
\usepackage{doi}
\usepackage{subcaption}
\usepackage{tikz}
\usetikzlibrary{math,patterns,positioning}
\usepackage{pgfplots}
\usepgfplotslibrary{fillbetween}
\pgfplotsset{compat=1.14}
\usepackage{datetime}\usdate
\usepackage{graphicx}

\usepackage{nicefrac}
\usepackage{enumitem}

\usepackage{microtype}

\newtheorem{theorem}{Theorem}
\newtheorem{corollary}{Corollary}
\newtheorem{lemma}{Lemma}

\theoremstyle{definition}

\def \R{\mathbb{R}}
\newcommand{\nbb}{\mathbb{N}}
\newcommand{\qbb}{\mathbb{Q}}

\newcommand{\vecc}[1]{\ensuremath{\mathbf{#1}}}

\usepackage{xparse,mleftright}

\newcommand{\union}{\cup}
\newcommand{\inters}{\cap}
\newcommand{\map}{\longrightarrow}

\NewDocumentCommand\xDeclarePairedDelimiter{mmm}
 {%
  \NewDocumentCommand#1{som}{%
   \IfNoValueTF{##2}
    {\IfBooleanTF{##1}{#2##3#3}{\mleft#2##3\mright#3}}
    {\mathopen{##2#2}##3\mathclose{##2#3}}%
  }%
 }

\xDeclarePairedDelimiter\parens{(}{)}
\xDeclarePairedDelimiter\braces{\lbrace}{\rbrace}
\xDeclarePairedDelimiter\card{|}{|}
\xDeclarePairedDelimiter\floor{\lfloor}{\rfloor}

\DeclareDocumentCommand\set{om}{\braces{\IfNoValueF{#1}{#1 \,|\,} #2 }}

\DeclareDocumentCommand\resourceZero{m}{0_{#1}}
\DeclareDocumentCommand\resourceOne{m}{1_{#1}}
\DeclareDocumentCommand\costZero{m}{c_{\resourceZero{#1}}}
\DeclareDocumentCommand\costOne{m}{c_{\resourceOne{#1}}}

\makeatletter
\title{Existence and Complexity of Approximate Equilibria\\ in Weighted Congestion Games\thanks{A preliminary version of this paper appeared in ICALP'20~\cite{cggpw2020}.
\newline \noindent
Supported by the Alexander von Humboldt Foundation with funds from the German
    Federal Ministry of Education and Research (BMBF).
Part of this work was done
    while Y.~Giannakopoulos and D.~Poças were members of the Operations Research
    group at Technical University of Munich, School of Management.
D. Poças was also funded by FCT via LASIGE Research Unit, ref. UIDB/00408/2020 and ref. UIDP/00408/2020.
       }}

\author{
		George Christodoulou\thanks{School of Informatics, Aristotle University of Thessaloniki.
		Email:
		{\tt \href{mailto:gichristo@csd.auth.gr}{\nolinkurl{gichristo@csd.auth.gr}}
		}}
	\and
		Martin Gairing\thanks{Department of Computer Science, University of Liverpool.
		Email:
		{\tt \href{mailto:gairing@liverpool.ac.uk}{\nolinkurl{gairing@liverpool.ac.uk}}
		}}
	\and
		Yiannis Giannakopoulos\thanks{Department of Data Science, Friedrich-Alexander-Universität Erlangen-Nürnberg.
		Email:
		{\tt \href{mailto:yiannis.giannakopoulos@fau.de}{\nolinkurl{yiannis.giannakopoulos@fau.de}}}}
	\and
		Diogo Poças\thanks{LASIGE, Faculdade de Ciências, Universidade de Lisboa.
		Email: {\tt \href{mailto:dmpocas@fc.ul.pt}{\nolinkurl{dmpocas@fc.ul.pt}}}}
	\and
		Clara Waldmann\thanks{Operations Research Group, Technical University of Munich.
		Email:
		{\tt \href{mailto:clara.waldmann@tum.de}{\nolinkurl{clara.waldmann@tum.de}}}}
}

\date{March 7, 2022}

\begin{document}

\maketitle
\setcounter{page}{0}
\thispagestyle{empty}

\begin{abstract}
We study the existence of approximate pure Nash equilibria ($\alpha$-PNE) in
weighted atomic congestion games with polynomial cost functions of maximum degree
$d$.
Previously it was known that $d$-PNE always exist,
while nonexistence was established only for small constants, namely for $1.153$-PNE.
We improve significantly upon this gap, proving that such games in general do not have $\tilde{\varTheta}(\sqrt{d})$-PNE, which provides the first super-constant lower bound.

Furthermore, we provide a black-box gap-introducing method of combining such
nonexistence results with a specific circuit gadget, in order to derive
NP-completeness of the decision version of the problem. In particular, deploying
this technique we are able to show that deciding whether a weighted congestion game
has an $\tilde{O}(\sqrt{d})$-PNE is NP-complete. Previous hardness results were
known only for the special case of \emph{exact} equilibria and arbitrary cost
functions.

The circuit gadget is of independent interest and it allows us to also prove
hardness for a variety of problems related to the complexity of PNE in congestion
games. For example, we demonstrate that the question of existence of $\alpha$-PNE in
which a certain set of players plays a specific strategy profile is NP-hard for any
$\alpha<3^{\nicefrac{d}{2}}$, even for \emph{unweighted} congestion games.

Finally, we study the existence of approximate equilibria in weighted congestion
games with general (nondecreasing) costs, as a function of the number of players
$n$. We show that $n$-PNE always exist, matched by an almost tight nonexistence
bound of $\tilde\varTheta(n)$ which we can again transform into an NP-completeness
proof for the decision problem.
\end{abstract}
\clearpage

\section{Introduction}\label{sec:intro}

\emph{Congestion games} constitute the standard framework to study settings where
selfish players compete over common resources. They are one of the most well-studied
classes of games within the field of \emph{algorithmic game
theory}~\citep{Roughgarden2016,2007a}, covering a wide range of applications,
including, e.g., traffic routing and load balancing.
In their most general form, each player has her own weight and the latency on each
resource is a nondecreasing function of the total weight of players that occupy it.
The cost of a player on a given outcome is just the total latency that she is
experiencing, summed over all the resources she is using.

The canonical approach to analysing such systems and predicting the behaviour of the
participants is  the ubiquitous game-theoretic tool of equilibrium analysis.
More specifically, we are interested in the \emph{pure Nash equilibria (PNE)}  of
those games; these are stable configurations from which no player would benefit from
unilaterally deviating. However, it is a well-known fact that such desirable
outcomes might not always exist, even in very simple weighted congestion games. A natural
response, especially from a computer science perspective, is to relax the solution
notion itself by considering \emph{approximate} pure Nash equilibria ($\alpha$-PNE);
these are states from which, even if a player could improve her cost by deviating,
this improvement could not be by more than a (multiplicative) factor of $\alpha\geq
1$. Allowing the parameter $\alpha$ to grow sufficiently large, existence of $\alpha$-PNE is restored.
But how large does $\alpha$ really \emph{need} to be? And,
perhaps more importantly from a computational perspective, how hard is it to check
whether a specific game has indeed an $\alpha$-PNE?

\subsection{Related Work}\label{sec:related}

The origins of the systematic study of (atomic) congestion games can be traced back
to the influential work of~\citet{Rosenthal1973a}, who also  proved that
\emph{unweighted} congestion games always possess PNE. His proof is
based on a simple but ingenious \emph{potential function} argument, which up to this
day is essentially still the only general tool for establishing existence of pure equilibria.

In follow-up work~\citep{Goemans2005,Libman2001,Fotakis2005a}, the nonexistence of
PNE was demonstrated even for special simple classes of (weighted) games, including
network congestion games with quadratic cost functions and games where the player
weights are either $1$ or $2$. On the other hand, we know that equilibria do exist
for affine or exponential
latencies~\citep{Fotakis2005a,Panagopoulou2007,Harks2012a}, as well as for the class
of singleton\footnote{These are congestion games where the players can only occupy
single resources.} games~\citep{Fotakis2009a,Harks2012}. \citet{Dunkel2008} were
able to extend the nonexistence instance of~\citet{Fotakis2005a} to a gadget in order to show that deciding
whether a congestion game with step cost functions has a PNE is a (strongly) NP-hard problem, via a reduction from \textsc{3-Partition}.

Regarding approximate equilibria, \citet{Hansknecht2014} gave instances of very
simple, two-player polynomial congestion games that do not have $\alpha$-PNE, for $\alpha\approx
1.153$. This lower bound is achieved by numerically solving an optimization
program, using polynomial latencies of maximum degree $d=4$. On the positive side,
\citet{Caragiannis2011} proved that $d!$-PNE always exist; this upper bound on the
existence of $\alpha$-PNE was later improved to
$\alpha=d+1$~\citep{Hansknecht2014,cggs2018-journal} and
$\alpha=d$~\citep{Caragiannis:2019aa,gp2020}.

\subsection{Our Results and Techniques}\label{sec:results}

After formalizing our model in~\cref{sec:model}, in~\cref{sec:nonexistence} we show
the nonexistence of $\varTheta(\frac{\sqrt{d}}{\ln d})$-PNE for
polynomial congestion games of degree $d$. This is the first super-constant lower
bound on the nonexistence of $\alpha$-PNE, significantly improving upon the previous
constant of $\alpha\approx 1.153$ and reducing the gap with the currently best upper
bound of $d$. More specifically (\cref{th:nonexistence}), for any integer $d$ we
construct congestion games with polynomial cost functions of maximum degree $d$ (and
nonnegative coefficients) that do not have $\alpha$-PNE, for any $\alpha<\alpha(d)$
where $\alpha(d)$ is a function that grows as
$\alpha(d)=\varOmega\left(\frac{\sqrt{d}}{\ln d}\right)$.
To derive this bound, we had to use a novel construction with a number of players
growing unboundedly as a function of $d$.

Next, in~\cref{sec:circuit} we turn our attention to computational hardness
constructions. Starting from a Boolean circuit, we create a gadget that transfers
hard instances of the classic \textsc{Circuit Satisfiability} problem to (even
unweighted) polynomial congestion games. Our construction is inspired by the work of Skopalik and Vöcking~\cite{Skopalik2008}, who used a similar family of lockable circuit games in their PLS-hardness result. Using this gadget we can immediately
establish computational hardness for various computational questions of interest
involving congestion games (\cref{th:hardness_PNE}). For example, we show that
deciding whether a $d$-degree polynomial congestion game has an $\alpha$-PNE in
which a specific set of players play a specific strategy profile is NP-hard, even up
to exponentially-approximate equilibria; more specifically, the hardness holds for
\emph{any} $\alpha < 3^{\nicefrac{d}{2}}$.
Our investigation of the hardness questions presented in~\cref{th:hardness_PNE} (and
later on in~\cref{th:hardness_counting} as well)
was  inspired by some similar results presented before by Conitzer and Sandholm
\cite{Conitzer:2008aa} (and even earlier in~\cite{Gilboa1989}) for \emph{mixed} Nash
equilibria in general (normal-form) games.
To the best of our knowledge, our paper is the first to study these questions for
\emph{pure} equilibria in the context of congestion games.
It is of interest to also note here that our
hardness gadget is \emph{gap-introducing}, in the sense that the $\alpha$-PNE and exact PNE of the game coincide.

In~\cref{sec:hardness_existence} we demonstrate how one can combine the hardness
gadget of ~\cref{sec:circuit}, in a black-box way, with any nonexistence instance
for $\alpha$-PNE, in order to derive hardness for the decision version of the
existence of $\alpha$-PNE (\cref{th:bb_hardness}, \cref{th:hardness_existence}). As
a consequence, using the previous $\varOmega\left(\frac{\sqrt{d}}{\ln d}\right)$ lower bound
construction of~\cref{sec:nonexistence}, we can show that deciding whether a
(weighted) polynomial congestion has an $\alpha$-PNE is NP-hard, for any
$\alpha<\alpha(d)$, where $\alpha(d)=\varOmega\left(\frac{\sqrt{d}}{\ln d}\right)$
(\cref{th:hardness_existence2}). Since our hardness is established via a rather
transparent, ``master'' reduction from \textsc{Circuit Satisfiability}, which in
particular is parsimonious, one can derive hardness for a family of related
computation problems; for example, we show that computing the number of
$\alpha$-PNE of a weighted polynomial congestion game is \#P-hard
(\cref{th:hardness_counting}).

In~\cref{sec:general_costs} we drop the assumption on polynomial cost
functions, and study the existence of approximate equilibria under arbitrary
(nondecreasing) latencies as a function of the number of players $n$.
We prove that $n$-player congestion games always have $n$-PNE
(\cref{th:existence_general_n}). As a consequence, one cannot hope to derive
super-constant nonexistence lower bounds by using just simple instances with a fixed
number of players (similar to, e.g., Hansknecht et al.~\cite{Hansknecht2014}).
In particular, this shows that the super-constant number of players in
our construction in~\cref{th:nonexistence} is necessary.
Furthermore, we pair
this positive result with an almost matching lower bound
(\cref{th:nonexistence_general_costs}): we give examples of $n$-player congestion
games (where latencies are simple step functions with a single breakpoint) that
do not have $\alpha$-PNE for all $\alpha<\alpha(n)$, where $\alpha(n)$ grows
according to $\alpha(n)=\varOmega\left(\frac{n}{\ln n}\right)$.
Finally, inspired by our hardness construction for the
polynomial case, we also give a new reduction that establishes NP-hardness for deciding
whether an $\alpha$-PNE exists, for any $\alpha<\alpha(n)= \varOmega\left(\frac{n}{\ln n}\right)$.
Notice that now the number of players $n$ is part of the description of the game (i.e., part of the input) as opposed to the maximum degree $d$ for the polynomial case (which was assumed to be fixed).
On the other hand though, we have more flexibility on designing our gadget latencies, since they can be arbitrary functions.

Concluding, we would like to elaborate on a couple of points. First, the reader would
have already noticed that in all our hardness results the (in)approximability
parameter $\alpha$ ranges freely within an entire interval of the form
$[1,\tilde\alpha)$, where $\tilde \alpha$ is a function of the degree $d$ (for
polynomial congestion games) or of the number of players $n$; and that $\alpha$, $\tilde\alpha$ are
\emph{not} part of the problem's input. It is easy to see that
these features only make our results stronger, with respect to computational
hardness, but also more robust. Secondly, although in this introductory section all our hardness
results were presented in terms of NP-\emph{hardness}, they immediately translate to
NP-\emph{completeness} under standard assumptions on the parameter $\alpha$; e.g., if $\alpha$
is rational (for a more detailed discussion of this, see also the end
of~\cref{sec:model}).

\section{Model and Notation} \label{sec:model}

A (weighted, atomic) \emph{congestion game} is defined by: a finite (nonempty) set of
\emph{resources} $E$, each $e\in E$  having a nondecreasing
\emph{cost (or latency) function} $c_e:\R_{>0} \map \R_{\geq 0}$; and a finite (nonempty) set of
\emph{players} $N$, $\card{N}=n$, each $i\in N$ having a \emph{weight} $w_i>0$
and a set of \emph{strategies} $S_i\subseteq 2^E$.
If all players have the same weight, $w_i=1$ for all $i\in N$, the game is called
\emph{unweighted}.
A \emph{polynomial congestion game} of degree $d$, for $d$ a nonnegative integer, is a
congestion game such that all its cost functions are polynomials of degree at most $d$
with nonnegative coefficients.
A \emph{strategy profile} (or \emph{outcome}) $\vecc s=(s_1,s_2,\dots,s_n)$ is a
collection of strategies, one for each player, i.e.\ $\vecc s \in \vecc{S}=
S_1\times S_2
\times\dots\times S_n$. Each strategy profile $\vecc s$ induces a \emph{cost} of
$C_i(\vecc s)=\sum_{e\in s_i} c_e(x_e(\vecc s))$ to every player $i\in N$, where
$x_e(\vecc s)=\sum_{i: e\in s_i} w_i$ is the induced \emph{load} on resource $e$. An
outcome $\vecc s$ will be called \emph{$\alpha$-approximate (pure Nash) equilibrium
($\alpha$-PNE)}, where $\alpha\geq 1$, if no player can unilaterally improve her
cost by more than a factor of $\alpha$. Formally:
\begin{equation}
\label{eq:approx_PNE_def} C_i(\vecc s) \leq \alpha \cdot C_i\parens{s_i',\vecc s_{-i}}
\qquad \text{for all } i\in N \text{ and all } s_i'\in S_i.
\end{equation} Here we have used the standard game-theoretic notation of $\vecc
s_{-i}$ to denote the vector of strategies resulting from $\vecc s$ if we remove its
$i$-th coordinate; in that way, one can write $\vecc s=(s_i,\vecc s_{-i})$. Notice
that for the special case of $\alpha=1$, \eqref{eq:approx_PNE_def} is equivalent to
the classical definition of pure Nash equilibria; for emphasis, we will sometimes
refer to such $1$-PNE as \emph{exact} equilibria.

If \eqref{eq:approx_PNE_def} does not hold, it means that player $i$ could improve
her cost by more than $\alpha$ by moving from $s_i$ to some other strategy $s'_i$.
We call such a move \emph{$\alpha$-improving}. Finally, strategy $s_i$ is said to be
\emph{$\alpha$-dominating} for player $i$ (with respect to a fixed profile
$\vecc{s}_{-i}$) if

\begin{equation}
\label{eq:dominating_def} C_i \parens{s_i',\vecc s_{-i}} > \alpha \cdot C_i(\vecc s)
\qquad \text{ for all } s_i'\neq s_i.
\end{equation}
In other words, if a strategy $s_i$ is $\alpha$-dominating, every move from some other
strategy $s'_i$ to $s_i$ is $\alpha$-improving. Notice that each player $i$ can have at most one
$\alpha$-dominating strategy (for $\vecc{s}_{-i}$ fixed). In our proofs,
we will employ a \emph{gap-introducing} technique by constructing games with the
property that, for any player $i$ and any strategy profile $\vecc{s}_{-i}$, there is
always a (unique) $\alpha$-dominating strategy for player $i$. As a consequence, the
sets of $\alpha$-PNE and exact PNE coincide.

Finally, for a positive integer $n$, we will use $\Phi_n$ to denote the unique
positive solution of equation $(x+1)^n=x^{n+1}$. Then, $\Phi_n$ is strictly
increasing with respect to $n$, with $\Phi_1=\phi\approx 1.618$ (golden ratio) and
asymptotically $\Phi_n\sim \frac{n}{\ln n}$
(see~\citep[Lemma~A.3]{cggs2018-journal}).

\paragraph*{Computational Complexity}
Most of the results in this paper involve complexity questions, regarding the
existence of (approximate) equilibria.
We are interested in computational problems of the form: ``Given as input a game $G$, does it have an $\alpha$-PNE?'' for a fixed value of $\alpha$.
Whenever we deal with such statements, we
will implicitly assume that the congestion game instances given as inputs to our
problems can be succinctly represented in the following way:

\begin{itemize}[noitemsep]
\item all player have \emph{rational} weights;
\item the resource cost functions are
``efficiently computable''; for polynomial latencies in particular, we will assume that
the coefficients are \emph{rationals}; and for step functions we assume that their values
and breakpoints are \emph{rationals};
\item the strategy sets are
given \emph{explicitly}.
\end{itemize}

There are also computational considerations to be made about the number $\alpha$
appearing in the definition of $\alpha$-PNE. In our results (e.g.,
\cref{th:hardness_PNE,th:hardness_existence}), we will prove NP-hardness of
determining whether games have $\alpha$-PNE for any arbitrary real $\alpha$ below
the nonexistence bound, \emph{regardless of whether $\alpha$ is rational or
irrational, computable or uncomputable}. However, to prove NP-completeness, i.e.\ to
prove that the decision problem belongs in NP (as in \cref{th:hardness_existence}),
we need to be able to verify, given a strategy profile and a deviation of some
player, whether this deviation is an $\alpha$-improving move. This can be achieved
by additionally assuming that the \emph{upper Dedekind cut} of $\alpha$,
$R_\alpha=\set[q\in\qbb]{q>\alpha}$, is a language decidable in polynomial time. In
this paper we will refer to such an $\alpha$ as a \emph{polynomial-time computable}
real number. In particular, notice that rationals are polynomial-time computable;
thus the NP-completeness of the $\alpha$-PNE problem does hold for $\alpha$
rational. We refer the interested reader to \citet{Ko1983} for a detailed discussion
on polynomial-time computable numbers (which is beyond the scope of our paper), as
well as for a comparison with other axiomatizations using binary digits
representations or convergent sequences.
If, more generally, $\alpha:\nbb\map\R$ is a sequence of reals (as in
\cref{th:hardness_existence_general_alt}), we say that $\alpha$ is a
\emph{polynomial-time computable} real sequence if
$R_\alpha=\set[(n,q)\in \nbb \times \qbb]{q>\alpha(n)}$ is a language decidable in
polynomial time.

\section{The Nonexistence Gadget} \label{sec:nonexistence}

In this section we give examples of polynomial congestion games of degree $d$, that do \emph{not} have $\alpha(d)$-PNE; $\alpha(d)$ grows as $\varOmega\left(\frac{\sqrt{d}}{\ln d}\right)$.
Fixing a degree $d\geq 2$, we construct a family of games $\mathcal{G}^d_{(n,k,w,\beta)}$, specified by parameters $n \in \nbb, k \in \set{1, \ldots, d}, w \in [0,1]$, and $\beta \in [0,1]$.
 In $\mathcal{G}^d_{(n,k,w,\beta)}$ there are $n+1$ players: a \emph{heavy player} of weight $1$ and $n$ \emph{light players} $1,\ldots,n$ of equal weights $w$.
 There are $2(n+1)$ resources $a_0,a_1, \ldots, a_n, b_0,b_1,\ldots,b_n$ where $a_0$ and $b_0$ have the same cost function $c_0$ and all other resources $a_1,\ldots,a_n,b_1, \ldots, b_n$ have the same cost function $c_1$ given by
 \[ c_0(x)=x^k \quad \text{and} \quad c_1(x)=\beta x^d.
 \]
 Each player has exactly two strategies, and the strategy sets are given by
 \[ S_0=\set{ \set{a_0,\ldots,a_n},\set{b_0,\ldots,b_n} }
  \quad\text{and}\quad
  S_i=\set{\set{a_0,b_i},\set{b_0,a_i} }\quad\text{for }i=1,\ldots,n.
 \]
 The structure of the strategies is visualized in \Cref{fig:nonexistence}.

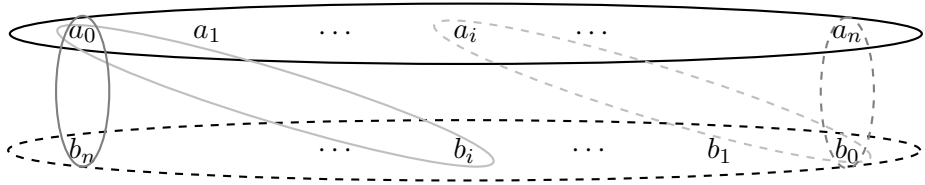
\begin{figure}[t]
	\centering
	\begin{tikzpicture}
		\def \d {1}
    \def \w {0.35}
		\draw node (a0) {$a_0$}
		 			node [right = \d of a0] (a1) {$a_1$}
					node [right = \d of a1] (adots1) {$\cdots$}
					node [right = \d of adots1] (ai) {$a_i$}
					node [right = \d of ai] (adots2) {$\cdots$}
          node [right = \d of adots2] (phb1) {\phantom{$b_1$}}
					node [right = \d of phb1] (an) {$a_n$}
					node [below right = 0.44*\d and 0.44*\d of a1] (mai) {}
					node [below = 0.4*\d of an] (man) {}
					;

		\draw node [below = \d of a0] (bn) {$b_n$}
          node [right = \d of bn] (pha1) {\phantom{$a_1$}}
					node [right = \d of pha1] (bdots1) {$\cdots$}
					node [right = \d of bdots1] (bi) {$b_i$}
					node [right = \d of bi] (bdots2) {$\cdots$}
					node [right = \d of bdots2] (b1) {$b_1$}
					node [below = \d of an] (b0) {$b_0$}
					node [above left = 0.35*\d and 0.45*\d of b1] (mbi) {}
					node [above = 0.35*\d of bn] (mbn) {}
					;

    \draw[thick] (ai) ellipse (6*\d cm and 0.4 cm);
		\draw[thick, dashed] (bi) ellipse (6*\d cm and 0.4 cm);

		\draw[rotate around = {-17 : (mai)}, thick, color=lightgray]
				(mai) ellipse (3*\d cm and \w cm);
		\draw[rotate around = {-17 : (mbi)}, thick, color=lightgray, dashed]
						(mbi) ellipse (3*\d cm and \w cm);
		\draw[rotate around ={90 : (man)}, thick, color=gray, dashed]
			(man) ellipse (\d cm and \w cm);
		\draw[rotate around = {90 : (mbn)}, thick, color=gray]
			(mbn) ellipse (\d cm and \w cm);

	\end{tikzpicture}
	\caption{Strategies of the game $\mathcal{G}^d_{(n,k,w,\beta)}$. Resources contained in the two ellipses of the same colour correspond to the two strategies of a player.
	The strategies of the heavy player and light players $n$ and $i$ are depicted in black, grey and light grey, respectively.}
	\label{fig:nonexistence}
\end{figure}

In the following theorem we give a lower bound on $\alpha$, depending on parameters
$(n,k,w,\beta)$, such that games $\mathcal{G}^d_{(n,k,w,\beta)}$ do not admit an
$\alpha$-PNE. Maximizing this lower bound over all games in the family, we obtain a
general lower bound $\alpha(d)$ on the inapproximability for polynomial congestion
games of degree $d$ (see~\eqref{eq:nonexist_ratio_opt} and its plot
in~\cref{fig:nonexistence_small}). Finally, choosing specific values for the
parameters $(n,k,w,\beta)$, we prove that $\alpha(d)$  is asymptotically lower
bounded by $\varOmega(\frac{\sqrt{d}}{\ln d})$.

\begin{theorem}
\label{th:nonexistence}
For any integer $d\geq 2$, there exist (weighted) polynomial congestion games of
degree $d$ that do not have $\alpha$-PNE for any $\alpha<\alpha(d)$,
where
\begin{samepage}
\begin{align}
\alpha(d)=&\sup_{n,k,w,\beta}\min \set{\frac{1+n\beta(1+w)^d}{(1+nw)^k+n\beta},\frac{(1+w)^k+\beta w^d}{(nw)^k+\beta(1+w)^d} } \label{eq:nonexist_ratio_opt}\\
&\text{s.t.}\quad n\in\nbb,k\in\{1,\ldots,d\},w\in[0,1],\beta\in[0,1].\nonumber
\end{align}
\end{samepage}

In particular, we have the asymptotics $\alpha(d)=\Omega\parens{\frac{\sqrt{d}}{\ln d}}$ and the bound $\alpha(d)\geq\frac{\sqrt{d}}{2\ln d}$, valid for large enough $d$. A plot of the exact values of $\alpha(d)$ (given
by~\eqref{eq:nonexist_ratio_opt}) for small degrees can be found
in~\Cref{fig:nonexistence_small}.
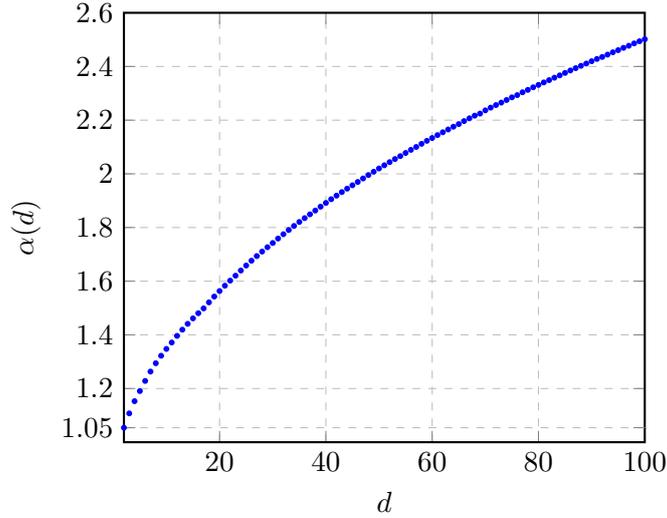
\begin{figure}[h]
\centering
\begin{tikzpicture}
\begin{axis}[
  xmin=2, xmax=100,
  ymin=1, ymax=2.6,
    ytick={1.2,1.4,...,2.6},
    ymajorgrids=true,
    xmajorgrids=true,
    extra y ticks={1.054},
    grid style=dashed,
    ylabel = {$\alpha(d)$},
    xlabel = {$d$},
    thick
]
\addplot[color=blue, only marks, mark size=0.7pt] coordinates
    {(2,1.0540141790979473396) (3,1.107370153362438117) (4,1.152889535629472211)
    (5,1.190508423862156880) (6,1.228303631508044132) (7,1.263116408601240858)
    (8,1.294214790723168450) (9,1.3222023252924773678) (10,1.347548347473695593)
    (11,1.3714041370263117508) (12,1.39618289713998503) (13,1.4193169771557038078)
    (14,1.440998681161281296) (15,1.4613811526582083621) (16,1.480567615135332816)
    (17,1.4986700525530704513) (18,1.521062799015134105) (19,1.54270838696362820)
    (20,1.563344248138408466) (21,1.583131179299605525) (22,1.602018079122111087)
    (23,1.620498874702096227) (24,1.639764463938681879) (25,1.658345416452749645)
    (26,1.676274654881171252) (27,1.693589577715126061) (28,1.7103216900329176963)
    (29,1.72648901720613499) (30,1.742297612260275997) (31,1.758910900141621082)
    (32,1.775000615409749991) (33,1.79061621987729414) (34,1.805776459710548517)
    (35,1.820643319156171548) (36,1.835039885996982471) (37,1.849051356777990172)
    (38,1.8632178093546148756) (39,1.877580446410725499) (40,1.891594176768753863)
    (41,1.905419217318211018) (42,1.918743797953874138) (43,1.93198155245319265)
    (44,1.944796558563581325) (45,1.957339942090330493) (46,1.969741285371048421)
    (47,1.9827320567687234043) (48,1.995396659292826213) (49,2.0077598678151992514)
    (50,2.019951919538817920) (51,2.0318506437685694077) (52,2.043557694254189955)
    (53,2.0550726405050166442) (54,2.065868165427579875) (55,2.0777540426080042548)
    (56,2.0893167741414560007) (57,2.100146688317447379) (58,2.111757748261389386)
    (59,2.1229607289450792557) (60,2.1334160715553645774) (61,2.1442905418729679244)
    (62,2.154793281438282600) (63,2.165243330097857295) (64,2.1755625648044519291)
    (65,2.18542146022641392) (66,2.196527068187740926) (67,2.206819681430450796)
    (68,2.216708579001054468) (69,2.2245048535808527924) (70,2.2364113021405588641)
    (71,2.246628630071107115) (72,2.256184841621078045) (73,2.26537600640903807)
    (74,2.2750610031117403033) (75,2.2849568000514343187) (76,2.293266798574380277)
    (77,2.303810861124610325) (78,2.312788452929739513) (79,2.322317954524468358)
    (80,2.331345711455386609) (81,2.340538676470352899) (82,2.349088085174637016)
    (83,2.358120028751895405) (84,2.366807583374904258) (85,2.3759335067794812244)
    (86,2.385072660331966221) (87,2.393628218527229176) (88,2.4024489095452699624)
    (89,2.410849597185745922) (90,2.419543602769812323) (91,2.428028013943385063)
    (92,2.435463573940856693) (93,2.444122869529486414) (94,2.452736829615797182)
    (95,2.4608168637707354936) (96,2.469384773292750051) (97,2.477212167857703696)
    (98,2.485892762433796739) (99,2.4938834934272180006) (100,2.5015225634873647209)
    };
\end{axis}
\end{tikzpicture}
\caption{Nonexistence of $\alpha(d)$-PNE for weighted polynomial
congestion games of degree $d$, as given by~\eqref{eq:nonexist_ratio_opt}
in~\cref{th:nonexistence}, for $d=2,3,\dots,100$.
In particular, for small values of $d$, $\alpha(2)\approx 1.054$, $\alpha(3) \approx 1.107$ and $\alpha(4)\approx 1.153$.
}
\label{fig:nonexistence_small}
\end{figure}
\end{theorem}

Interestingly, for the special case of $d=2,3,4$, the values of $\alpha(d)$ (see \cref{fig:nonexistence_small}) yield
\emph{exactly} the same lower bounds with Hansknecht et al.~\cite{Hansknecht2014}.
This is a direct consequence of the fact that $n=1$ turns out to be an optimal
choice in~\eqref{eq:nonexist_ratio_opt} for $d\leq 4$, corresponding to an instance
with only $n+1=2$ players (which is the regime of the construction
in~\cite{Hansknecht2014}); however, this is not the case for larger values of $d$, where more
players are now needed in order to derive the best possible value
in~\eqref{eq:nonexist_ratio_opt}.
Furthermore, as we discussed also in
\cref{sec:results}, no construction with only $2$ players can result in bounds
larger than $2$ (\cref{th:existence_general_n}).

\begin{proof}
Due to symmetries, it is enough to just consider the following two cases for the strategy profiles in game $\mathcal{G}^d_{(n,k,w,\beta)}$ described above:

\emph{\underline{Case 1:} The heavy player is alone on resource $a_0$.}
  This means that every light player $i\in\{1,\ldots, n\}$ must have chosen strategy $\{b_0,a_i\}$. Thus the heavy player incurs a cost of $c_0(1)+n c_1(1+w)$; while, deviating to strategy $\{b_0,\ldots,b_n\}$, she would incur a cost of $c_0(1+nw)+nc_1(1)$. The improvement factor can then be lower bounded by
	\begin{equation*}
    \frac{c_0(1)+n c_1(1+w)}{c_0(1+nw)+nc_1(1)}=\frac{1+n\beta(1+w)^d}{(1+nw)^k+n\beta}.
    \label{eq:nonexist_ratio_1}
  \end{equation*}

\emph{\underline{Case 2:} The heavy player shares resource $a_0$ with at least one light player $i\in\{1,\ldots, n\}$.}
  Thus player $i$ incurs a cost of at least $c_0(1+w)+c_1(w)$; while, deviating to strategy $\{b_0,a_i\}$, she would incur a cost of at most $c_0(nw)+c_1(1+w)$. The improvement factor can then be lower bounded by
	\begin{equation*}\label{eq:nonexist_ratio_2}\frac{c_0(1+w)+c_1(w)}{c_0(nw)+c_1(1+w)}=\frac{(1+w)^k+\beta w^d}{(nw)^k+\beta(1+w)^d}.\end{equation*}

In order for the game to not have an $\alpha$-PNE, it is enough to guarantee that both ratios are greater than $\alpha$.
Maximizing these ratios over all games in the family, yields the lower bound in the statement of the theorem,

\begin{align*}
\alpha(d)= &\sup_{n,k,w,\beta}\min
    \set{ \frac{1+n\beta(1+w)^d}{(1+nw)^k+n\beta},\frac{(1+w)^k+\beta w^d}{(nw)^k+\beta(1+w)^d} }\\
    &\text{s.t.} \quad n\in\nbb,k\in\{1,\ldots,d\},w\in[0,1],\beta\in[0,1].
\end{align*}

For small values of $d$ the above quantity can be computed numerically
(see~\cref{fig:nonexistence_small}); in particular, for $d=2,3,4$ this yields the
same lower bounds as in~\citet{Hansknecht2014}, since  $n=1$ is the optimal choice.

Next we prove the asymptotics $\alpha(d)=\Omega\left(\frac{\sqrt{d}}{\ln d}\right)$. To that end, we take the following choice of parameters:
\[
  w=\frac{\ln d}{2d}\,,\,
  k=\left\lceil\frac{\ln d}{2\ln\ln d}\right\rceil\,,\,
  \beta=\frac{1}{d^{\frac{k}{2(k+1)}}(1+w)^d}\,,\,
  n=\left\lfloor\frac{1}{d^{\frac{1}{2(k+1)}}w}\right\rfloor.
\]
One can check that this choice satisfies $k \in \set{1, \ldots, d}$ (for $d \ge 4$) and $w,\beta \in [0,1]$.
We can bound the expressions appearing in \eqref{eq:nonexist_ratio_opt} as follows.

\begin{align}
  1+n\beta(1+w)^d
  &\geq1+\left(\frac{1}{d^{\frac{1}{2(k+1)}}w}-1\right)\frac{1}{d^{\frac{k}{2(k+1)}}(1+w)^d}(1+w)^d\nonumber\\
  &= 1+\left(\frac{2d}{d^{\frac{1}{2(k+1)}}\ln d}-1\right)\frac{1}{d^{\frac{k}{2(k+1)}}}\nonumber\\
  &=\frac{2d}{d^{\frac{1}{2(k+1)}+\frac{k}{2(k+1)}}\ln d}+1-\frac{1}{d^{\frac{k}{2(k+1)}}}\nonumber\\
  &\geq\frac{2d}{d^{\nicefrac{1}{2}}\ln d}\tag{since $d\geq1$}\\
  &=\frac{2\sqrt{d}}{\ln d};\label{eq:nonexist_aux_bound_1}\\
  (1+nw)^k+n\beta
  &\leq\left(1+\frac{1}{d^{\frac{1}{2(k+1)}}w}w\right)^k + \frac{1}{d^{\frac{1}{2(k+1)}}wd^{\frac{k}{2(k+1)}}(1+w)^d} \nonumber\\
  &=\left(1+d^{-\frac{1}{2(k+1)}}\right)^k+\frac{1}{d^{\nicefrac{1}{2}}\frac{\ln d}{2d}\left(1+\frac{\ln d}{2d}\right)^d}\nonumber\\
  &=\left(1+d^{-\frac{1}{2(k+1)}}\right)^k+\frac{2\sqrt{d}}{\ln d\left(1+\frac{\ln d}{2d}\right)^d};\label{eq:nonexist_aux_bound_2}\\
  (1+w)^k+\beta w^d&\geq 1; \label{eq:nonexist_aux_bound_3}\\
  (nw)^k+\beta(1+w)^d
  &\leq \left(\frac{1}{d^{\frac{1}{2(k+1)}}w}w\right)^k+\frac{1}{d^{\frac{k}{2(k+1)}}(1+w)^d}(1+w)^d\nonumber\\
  &=2\cdot\frac{1}{d^{\frac{k}{2(k+1)}}}=\frac{2d^{\frac{1}{2(k+1)}}}{\sqrt{d}}\leq\frac{2d^\frac{\ln\ln d}{\ln d}}{\sqrt{d}}=\frac{2\ln d}{\sqrt{d}}. \label{eq:nonexist_aux_bound_4}
\end{align}
In the Appendix, we prove (\cref{lem:nonexist_ratio_2_global}) that the final quantity in \eqref{eq:nonexist_aux_bound_2} converges to 1 as $d\rightarrow\infty$; in particular, it is upper bounded by $4$ for $d$ large enough (numerically, we can observe that $d\geq 8$ suffices). Thus, we can lower bound the ratios of \eqref{eq:nonexist_ratio_opt} as

\begin{align}
  \frac{1+n\beta(1+w)^d}{(1+nw)^k+n\beta}&\geq\frac{\frac{2\sqrt{d}}{\ln d}}{4}=\frac{\sqrt{d}}{2\ln d}=\Omega\left(\frac{\sqrt{d}}{\ln d}\right), \tag{from \eqref{eq:nonexist_aux_bound_1}, \eqref{eq:nonexist_aux_bound_2} and large $d$}\\
  \frac{(1+w)^k+\beta w^d}{(nw)^k+\beta(1+w)^d}&\geq\frac{1}{\frac{2\ln d}{\sqrt{d}}}=\frac{\sqrt{d}}{2\ln d}=\Omega\left(\frac{\sqrt{d}}{\ln d}\right).\tag{from \eqref{eq:nonexist_aux_bound_3} and \eqref{eq:nonexist_aux_bound_4}}
\end{align}

This proves the asymptotics and the bound $\alpha(d)\geq\frac{\sqrt{d}}{2\ln d}$ for large $d$.
\end{proof}

\section{The Hardness Gadget} \label{sec:circuit}

In this section we construct an unweighted polynomial congestion game from a Boolean
circuit. In the $\alpha$-PNE of this game the players emulate the computation of the
circuit. This gadget will be used in reductions from \textsc{Circuit Satisfiability}
to show NP-hardness of several problems related to the existence of approximate
equilibria with some additional properties. For example, deciding whether a
congestion game has an $\alpha$-PNE where a certain set of players choose a specific
strategy profile (\cref{th:hardness_PNE}).

\paragraph*{Circuit Model}
We consider Boolean circuits consisting of NOT gates and 2-input NAND gates only. We assume that the two inputs to every NAND gate are different. Otherwise we replace the NAND gate by a NOT gate, without changing the semantics of the circuit.
We further assume that every input bit is connected to exactly one gate and this gate is a NOT gate. See \Cref{subfig:validcirc} for a \emph{valid} circuit.
In a valid circuit we replace every NOT gate by an equivalent NAND gate, where one of the inputs is fixed to 1.
See the replacement of gates $g_5, g_4$ and $g_2$ in the example in \Cref{subfig:elimNotgates}.
Thus, we look at circuits of 2-input NAND gates where both inputs to a NAND gate are different and every input bit of the circuit is connected to exactly one NAND gate where the other input is fixed to 1. A circuit of this form is said to be in \emph{canonical form}.
For a circuit $C$ and a vector $x \in \set{0,1}^n$ we denote by $C(x)$ the output of the circuit on input $x$.

We model a circuit $C$ in canonical form as a \emph{directed acyclic graph}. The nodes of this graph correspond to the input bits $x_1, \ldots, x_n$, the gates $g_1, \ldots, g_K$ and a node $1$ for all fixed inputs.
There is an arc from a gate $g$ to a gate $g'$ if the output of $g$ is input to gate $g'$ and there are arcs from the fixed input and all input bits to the connected gates.
We index the gates in reverse topological order, so that all successors of a gate $g_k$ have a smaller index and the output of gate $g_1$ is the output of the circuit.
Denote by $\delta^+(v)$ the set of the direct successors of node $v$.
Then we have $|\delta^+(x_i)| = 1$ for all input bits $x_i$ and $\delta^+(g_k) \subseteq \set[g_{k'}]{k' < k} $ for every gate $g_k$.
See \Cref{fig:circ} for an example of a valid circuit, its canonical form and the corresponding directed acyclic graph.

\begin{figure}[t]
  \begin{subfigure}{0.33 \textwidth}
    \centering
   \includegraphics[trim={0.6cm 24.7cm 0 0}, clip, scale = 1.1]{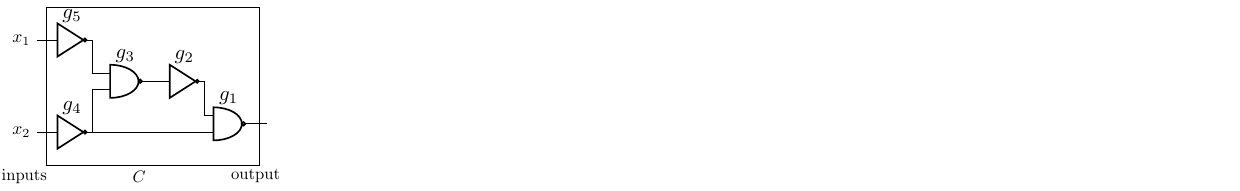}
  \caption{valid circuit $C$}
  \label{subfig:validcirc}
  \end{subfigure} \quad
  \begin{subfigure}{0.3 \textwidth}
    \centering
    \includegraphics[trim={0.6cm 24.7cm 14cm 0}, clip, scale = 1.1]{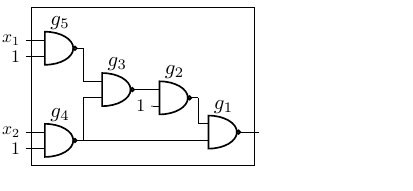}
  \caption{ canonical form of $C$}
  \label{subfig:elimNotgates}
  \end{subfigure} \qquad
  \begin{subfigure}{0.3 \textwidth}
    \centering
    \includegraphics[scale=1.15]{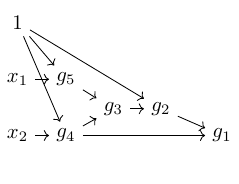}
  \caption{directed acyclic graph}
  \label{subfig:circuitdag}
  \end{subfigure}
\caption{Example of a valid circuit $C$ (having both NOT and NAND gates), its canonical form (having only NAND gates), and the directed acyclic graph corresponding to $C$.}
\label{fig:circ}
\end{figure}

\paragraph*{Translation to Congestion Game}
Fix some integer $d \ge 1$ and a parameter $\mu \ge 1 + 2\cdot3^{d + \nicefrac{d}{2}}$.
From a valid circuit in canonical form with input bits $x_1, \ldots, x_n$, gates $g_1, \ldots, g_K$ and the extra input fixed to $1$, we construct a polynomial congestion game $\mathcal{G}_\mu^d$ of degree $d$.
There are $n$ \emph{input players} $X_1, \ldots, X_n$ for every input bit, a \emph{static player} $P$ for the input fixed to $1$, and $K$ \emph{gate players} $G_1, \ldots, G_K$ for the output bit of every gate. $G_1$ is sometimes called \emph{output player} as $g_1$ corresponds to the output $C(x)$.

The idea is that every input and every gate player have a \emph{zero} and a \emph{one strategy}, corresponding to the respective bit being $0$ or $1$. In every $\alpha$-PNE we want the players to emulate the computation of the circuit, i.e.~the NAND semantics of the gates should be respected.
For every gate $g_k$, we introduce two \emph{resources} $\resourceZero{k}$ and $\resourceOne{k}$. The zero (one) strategy of a player consists of the $\resourceZero{k'}$ ($\resourceOne{k'}$) resources of the direct successors in the directed acyclic graph corresponding to the circuit and its own $\resourceZero{k}$ ($\resourceOne{k}$) resource (for gate players). The static player has only one strategy playing all $\resourceOne{k'}$ resources of the gates where one input is fixed to 1.

Formally, we have $S_{X_i} = \set{s_{X_i}^0, s_{X_i}^1 }$, where
\begin{equation*}
s_{X_i}^0 = \set[\resourceZero{k}]{g_k \in \delta^+(x_i)} \text{ and } s_{X_i}^1 = \set[\resourceOne{k}]{g_k \in \delta^+(x_i)}
\end{equation*}
are the zero and one strategy of input player $X_i$. Recall that $\delta^+(x_i)$ is the set of direct successors of $x_i$, thus every strategy of an input player consists of exactly one resource.
For a gate player $G_k$ we have $S_{G_k} = \set{s_{G_k}^0, s_{G_k}^1 }$ with the two strategies
\begin{equation*}
  s_{G_k}^0 = \set{\resourceZero{k}} \union \set[\resourceZero{k'}]{g_{k'} \in \delta^+(g_k)} \text{ and } s_{G_k}^1 = \set{\resourceOne{k}} \union \set[\resourceOne{k'}]{g_{k'} \in \delta^+(g_k)}
\end{equation*}
consisting of at most $k$ resources each.
The strategy of the static player is $s_P = \set[\resourceOne{k}]{g_k \in \delta^+(1)}$.
Notice that all $3$ players related to a gate $g_k$ (gate player $G_k$ and the two players corresponding to the input bits) are different and observe that every resource $\resourceZero{k}$ and $\resourceOne{k}$ can be played by exactly those $3$ players.

We define the cost functions of the resources using parameter $\mu$.
The cost functions for resources $\resourceOne{k}$ are given by $\costOne{k}$ and for resources $\resourceZero{k}$ by $\costZero{k}$, where
\begin{equation}
  \costOne{k}(x) = \mu^k x^d \qquad \text{and} \qquad
  \costZero{k}(x) = \lambda \mu^k x^d
  \text{, with }
 \lambda = 3^{\nicefrac{d}{2}}
 .
 \label{eq:lambda_def}
 \end{equation}

Our construction here is inspired by the lockable circuit games of Skopalik and
Vöcking~\cite{Skopalik2008}. The key technical differences are that our gadgets use
polynomial cost functions (instead of general cost functions) and only $2$ resources
per gate (instead of $3$). Moreover, while in~\cite{Skopalik2008} these games are used as
part of a PLS-reduction from \textsc{Circuit/FLIP}, we are also interested in
constructing a gadget to be studied on its own, since this can give rise to
additional results of independent interest (see~\cref{th:hardness_PNE}).

\paragraph*{Properties of the Gadget}
For a valid circuit $C$ in canonical form consider the game $\mathcal{G}^d_\mu$ as defined above.
We interpret any strategy profile $\vecc{s}$ of the input players as a bit vector $x \in \set{0,1}^n$ by setting $x_i = 0$ if $s_{X_i} = s_{X_i}^0$ and $x_i = 1 $ otherwise.
The gate players are said to \emph{follow the NAND semantics} in a strategy profile, if for every gate $g_k$ the following holds:
\begin{itemize}
\item if both players corresponding to the input bits of $g_k$ play their one strategy, then the gate player $G_k$ plays her zero strategy;
\item if at least one of the players corresponding to the input bits of $g_k$ plays her zero strategy, then the gate player $G_k$ plays her one strategy.
\end{itemize}
We show that for the right choice of $\alpha$, the set of $\alpha$-PNE in $\mathcal{G}^d_\mu$ is the same as the set of all strategy profiles where the gate players follow the NAND semantics.

Define
\begin{equation}
  \label{eq:epsilon_mu_def}
  \varepsilon(\mu) = \frac{3^{d + \nicefrac{d}{2}}}{\mu - 1}.
\end{equation}
From our choice of $\mu$, we obtain $ 3^{\nicefrac{d}{2}} - \varepsilon(\mu)
\ge 3^{\nicefrac{d}{2}} - \frac{1}{2} > 1 $.
For any valid circuit $C$ in canonical form and a valid choice of $\mu$ the following lemma holds for $\mathcal{G}^d_\mu$.
\begin{lemma}
\label{lem:circuit_NANDsemantics}
Let $\vecc{s}_X$ be any strategy profile for the input players $X_1, \ldots, X_n$ and let $x \in \set{0,1}^n$ be the bit vector represented by $\vecc{s}_X$.
For any $\mu \ge 1+2\cdot3^{d+\nicefrac{d}{2}}$ and any $1 \le \alpha < 3^{\nicefrac{d}{2}} -
\varepsilon(\mu)$, there is a unique $\alpha$-PNE\footnote{Which, as a
matter of fact, is actually also an \emph{exact} PNE.} in $\mathcal{G}^d_\mu$ where the input
players play according to $\vecc{s}_X$. In particular, in this $\alpha$-PNE the gate players follow the NAND semantics, and the output player $G_1$ plays
according to $C(x)$.
\end{lemma}
\begin{proof}
  Let $\mu>1+2\cdot3^{d+\nicefrac{d}{2}}$ and $\alpha < 3^{\nicefrac{d}{2}} - \varepsilon(\mu)$.
  First, we fix the input players to the strategies given by $\vecc{s}_X$ and show that in any $\alpha$-PNE every gate player follows the NAND semantics, as otherwise changing to the strategy corresponding to the NAND of its input bits is an $\alpha$-improving move.
  Second, we show that in any $\alpha$-PNE where the gate players follow the NAND semantics, the input players have no incentive to change their strategy.
  In total we get that every strategy profile for the input players can be extended to an $\alpha$-PNE, where the gate players emulate the circuit. Hence, for fixed strategies of the input players given by $\vecc{s}_X$ this $\alpha$-PNE is unique.

  Let $\vecc{s}_X$ be any strategy profile for the input players $X_1, \ldots, X_n$ and let $\vecc{s}$ be an $\alpha$-PNE of $\mathcal{G}^d_\mu$ where the input players play according to $\vecc{s}_X$.
  Take $G_k$ to be any of the gate players and let $P_a$ and $P_b$ be the players corresponding to the input bits of gate $g_k$. Note that $P_a$ and $P_b$ can be other gate players or input players, and one of them can be the static player.
  To show that $G_k$ follows the NAND semantics we consider two cases.

  \emph{\underline{Case 1:} Both $P_a$ and $P_b$ play their one strategy in $\vecc{s}$.}
  As both $P_a$ and $P_b$ play resource $\resourceOne{k}$ and all three players $P_a, P_b$ and $G_k$ are different, the cost of $G_k$'s one strategy is at least $\costOne{k}(3)$.
  The cost of $G_k$'s zero strategy is at most $\costZero{k}(1) + \sum_{k'=1}^{k-1} \costZero{k'}(3)$.
  Thus, we have
  \begin{equation*}
    \frac{C_{G_k}(s_{G_k}^1, \vecc{s}_{-G_k})}{C_{G_k}(s_{G_k}^0, \vecc{s}_{-G_k})}
    \ge \frac{\costOne{k}(3)}{\costZero{k}(1) + \sum_{k'=1}^{k-1} \costZero{k'}(3)}
    = \frac{\mu^k 3^d}{\lambda \mu^k + \sum_{k'=1}^{k-1} \lambda \mu^{k'} 3^d}
    > \frac{3^d}{\lambda} \parens{ \frac{1}{1 + \frac{1}{\mu - 1} 3^d } },
  \end{equation*}
  where we used that $ \frac{1}{\mu^k} \sum_{k'=1}^{k-1} \mu^{k'} = \frac{1}{\mu^k} \parens{\frac{\mu ^ k - \mu}{\mu - 1}}  < \frac{1}{\mu - 1} $.
  By the definition of $\lambda$ (see~\eqref{eq:lambda_def}) and $\varepsilon(\mu)$ (see~\eqref{eq:epsilon_mu_def}), we obtain
  \begin{equation}\label{eq:nandsemantics_aux_1}
     \frac{3^d}{\lambda} \parens{ \frac{1}{1 + \frac{1}{\mu - 1} 3^d } }
    = 3^{\nicefrac{d}{2}} \parens{\frac{1}{1 + \frac{1}{\mu - 1} 3^d } }
    > 3^{\nicefrac{d}{2}} \parens{1 - \frac{1}{\mu - 1} 3^d  }
    = 3^{\nicefrac{d}{2}} - \varepsilon(\mu)
    > \alpha
    .
  \end{equation}
  Hence, changing from the one to the zero strategy would be an $\alpha$-improving move for $G_k$.
  Thus, $G_k$ must follow the NAND semantics and play her zero strategy in $\vecc{s}$.

  \emph{\underline{Case 2:} At least one of $P_a$ or $P_b$ is playing her zero strategy in $\vecc{s}$.}
  By similar arguments to the previous case, we obtain that the cost of $G_k$'s zero strategy is at least $\costZero{k}(2)$ and the cost of the one strategy is at most
  $ \costOne{k}(2) + \sum_{k'=1}^{k-1} \costOne{k'}(3) $.
  Then, we get that
  \begin{equation*}
    \frac{C_{G_k}(s_{G_k}^0, \vecc{s}_{-G_k})}{C_{G_k}(s_{G_k}^1, \vecc{s}_{-G_k})}
    \ge \frac{\costZero{k}(2)}{\costOne{k}(2) + \sum_{k'=1}^{k-1} \costOne{k'}(3)}
    = \frac{\lambda \mu^k 2^d}{\mu^k 2^d + \sum_{k'=1}^{k-1} \mu^{k'} 3^d}
    > \lambda \parens{ \frac{1}{1 + \frac{1}{\mu - 1} \parens{\frac{3}{2}}^d } }
    .
  \end{equation*}
  By the definition of $\lambda$ and $\varepsilon(\mu)$, we obtain
  \begin{equation}\label{eq:nandsemantics_aux_2}
     \lambda \parens{ \frac{1}{1 + \frac{1}{\mu - 1} \parens{\frac{3}{2}}^d } }
    > 3^{\nicefrac{d}{2}} \parens{\frac{1}{1 + \frac{1}{\mu - 1} 3^d } }
    > 3^{\nicefrac{d}{2}} \parens{1 - \frac{1}{\mu - 1} 3^d  }
    = 3^{\nicefrac{d}{2}} - \varepsilon(\mu)
    >\alpha
    .
  \end{equation}
  Hence, changing from the zero to the one strategy would be an $\alpha$-improving move for $G_k$.
  Thus, $G_k$ must follow the NAND semantics and play her one strategy in $\vecc{s}$.

  We just showed that, in an $\alpha$-PNE, every gate player must follow the NAND semantics.
  This implies that there is \emph{at most one} $\alpha$-PNE where the input players play according to $\vecc{s}_X$, since the NAND semantics uniquely define the strategy of the remaining players.
  To conclude the proof, we must argue that this yields in fact an $\alpha$-PNE, meaning that the input players are also `locked' to their strategies in $\vecc{s}_X$ and have no incentive to deviate.
  To that end, let $\vecc{s}$ be a strategy profile PNE of $\mathcal{G}^d_\mu$ where the gate
  players follow the NAND semantics and let $X_i$ be any of the input players.
  Recall that every input bit $x_i$ is connected to exactly one gate, say
  $g_k$, while the other input is fixed to 1. To show that $X_i$ does not
  have an incentive to change her strategy, we consider two cases.

  \emph{\underline{Case 1:} $X_i$ plays her one strategy in $\vecc{s}$.}
  As $G_k$ follows the NAND semantics in $\vecc{s}$ and the other input of $g_k$ is fixed to $1$, we know that $G_k$ must be playing her zero strategy.
  Thus, incurring a cost of $\costOne{k}(2) = \mu^k 2^d$ for $X_i$.
  On the other hand, if $X_i$ changed to her zero strategy this would incur a cost of $\costZero{k}(2) = \lambda \mu^k 2^d $.
  Since $\lambda = 3^{\nicefrac{d}{2}}> 3^{\nicefrac{d}{2}} - \varepsilon(\mu) > \alpha$, staying at her one strategy is $\alpha$-dominating for $X_i$.

  \emph{\underline{Case 2:} $X_i$ plays her zero strategy in $\vecc{s}$.}
  As $G_k$ follows the NAND semantics in $\vecc{s}$ and the other input of $g_k$ is fixed to $1$, we know that $G_k$ must be playing her one strategy.
  This incurs a cost of $\costZero{k}(1) = \lambda \mu^k$ to $X_i$.
  On the other hand, if $X_i$ changed to her one strategy this would incur a cost of $\costOne{k}(3) = \mu^k 3^d$.
  Again, it is not $\alpha$-improving for $X_i$ to change her strategy, as $\frac{3^d}{\lambda} = 3^{\nicefrac{d}{2}} > 3^{\nicefrac{d}{2}} - \varepsilon(\mu) > \alpha$.
\end{proof}

We are now ready to show our main result of this section; using the circuit game
described above, we show NP-hardness of deciding whether approximate equilibria with
additional properties exist.
 As already mentioned in the introduction, we emphasize here that for each of the problems in \Cref{th:hardness_PNE}, parameters $d$, $\alpha$ and $z$ are fixed (i.e., not part of the input).
\begin{theorem}
\label{th:hardness_PNE}
The following problems are NP-hard, even for \emph{unweighted} polynomial congestion games of degree $d\geq 1$, for all $\alpha\in [1,3^{\nicefrac{d}{2}})$ and all $z>0$:
\begin{itemize}
\item ``Does there exist an $\alpha$-PNE in which a certain subset of players is playing a specific strategy profile?''
\item ``Does there exist an $\alpha$-PNE in which a certain resource is used by at least one player?''
\item ``Does there exist an $\alpha$-PNE in which a certain player has cost at most $z$?''
\end{itemize}
\end{theorem}

\begin{proof}
  For the first problem we reduce from \textsc{Circuit Satisfiability}: given a Boolean circuit with $n$ input bits and one output bit, is there an assignment of the input bits where the output of the circuit is $1$?
  This problem is NP-hard even for circuits consisting only of 2-input NAND gates \citep{PapadimitriouComplexity}.
  Let $C'$ be a Boolean circuit of 2-input NAND gates.
  We transform $C'$ into a valid circuit $C$ by connecting every input bit to a NOT gate and the output of this NOT gate to all gates connected to the input bit in $C'$.
  Thus, $C'(x) = C(\bar{x})$, where $\bar{x}$ denotes the vector obtained from $x \in \set{0,1}^n$ by flipping every bit.
  Hence, we have that $C'$ is a YES-instance to \textsc{Circuit Satisfiability} if and only if $C$ is a YES-instance.

  Let $\alpha \in [1,3^{\nicefrac{d}{2}})$, then there is an $\varepsilon > 0$ with $\alpha < 3^{\nicefrac{d}{2}} - \varepsilon$. We set $\mu =1 + \frac{3^{d + \nicefrac{d}{2}}}{ \min \set{\varepsilon,1}}$.
  For this choice of $\mu$, we obtain $\varepsilon(\mu) \le \varepsilon$ and thus $3^{\nicefrac{d}{2}} - \varepsilon(\mu) \ge 3^{\nicefrac{d}{2}} - \varepsilon > \alpha$.
  From the canonical form of $C$ we construct\footnote{To be precise, the description of the game $\mathcal{G}^d_\mu$ involves the quantities $\mu=1 + \frac{3^{d + \nicefrac{d}{2}}}{ \min \set{\varepsilon,1}}$ and $\lambda=3^{\nicefrac{d}{2}}$, which in general might be irrational. In order to incorporate this game into our reduction, it is enough to take a rational $\mu$ such that $\mu>1 + \frac{3^{d + \nicefrac{d}{2}}}{ \min \set{\varepsilon,1}}$,
  and a rational $\lambda$ such that $\alpha\left(1+\frac{1}{\mu-1}3^d\right)<\lambda<\frac{3^d}{\alpha\left(1+\frac{1}{\mu-1}3^d\right)}$. In this way, $\mathcal{G}^d_\mu$ is described entirely via rational numbers, while preserving the inequalities in \eqref{eq:nandsemantics_aux_1} and \eqref{eq:nandsemantics_aux_2}.} the game $\mathcal{G}^d_\mu$.
  The subset of players we are looking at is the output player $G_1$ and the specific strategy for $G_1$ is her one strategy $s_{G_1}^1$.
  We show that there is an $\alpha$-PNE where $G_1$ plays $s_{G_1}^1$ if and only if $C$ is a YES-instance to \textsc{Circuit Satisfiability}.

  Suppose there is a bit vector $x \in \set{0,1}^n$ such that $C(x) = 1$.
  Let $\vecc{s}_X$ be the strategy profile for the input players of $\mathcal{G}^d_\mu$ corresponding to $x$.
  Since $3^{\nicefrac{d}{2}} - \varepsilon(\mu) > \alpha$, \Cref{lem:circuit_NANDsemantics} holds for $\mathcal{G}^d_\mu$ and $\alpha$.
  Hence, the profile $\vecc{s}_X$ can be extended to an $\alpha$-PNE where $G_1$ plays according to $C(x)$. Thus, there is an $\alpha$-PNE where $G_1$ plays $s_{G_1}^1$.

  On the other hand, suppose for all bit vectors $x \in \set{0,1}^n$ it holds $C(x) = 0$.
  Again, by \Cref{lem:circuit_NANDsemantics} we know that for any choice of strategies for the input players, the only $\alpha$-PNE is a profile where the gate players follow the NAND semantics.
  Thus in this case, $G_1$ is playing $s_{G_1}^0$ in any $\alpha$-PNE.

  To show NP-hardness of the second problem, simply apply the same reduction as in the first problem, but now include an additional resource of cost $0$ to the strategy $s^1_{G_1}$ of the output player $G_1$. Now $C$ is a YES-instance to \textsc{Circuit Satisfiability} iff the additional resource is used by some player.


  The hardness of the third problem is shown by a reduction from \textsc{Circuit Satisfiability}, similar to the proof for the first problem. For $\alpha \in [1, 3^{\nicefrac{d}{2}})$ we choose $\mu$ as before, so that \Cref{lem:circuit_NANDsemantics} holds for $\alpha$ and the game $\mathcal{G}_\mu^d$ for a suitable circuit.
  Let $C'$ be an instance of \textsc{Circuit Satisfiability}. By negating the output of $C'$, we obtain a circuit $\overline{C'}$. As before we transform $\overline{C'}$ to a valid circuit $\overline{C}$, so that $C'(x) = \lnot \overline{C}(\bar{x})$ holds. From the canonical form of $\overline{C}$ we construct the game $\mathcal{G}^d_\mu$. Note that the output player $G_1$ of this game is the output of a gate, where one of the inputs is fixed to $1$, as we negated the output of $C'$ by connecting the output of $C'$ to a NOT gate.
  We show that there is an $\alpha$-PNE in $\mathcal{G}^d_\mu$, where $G_1$ has cost at most $\lambda \mu$, if and only if $C'$ is a YES-instance to \textsc{Circuit Satisfiability}.

  Suppose there is a bit vector $x \in \set{0,1}^n$ with $C'(x) = 1$, then there is a vector $\bar{x}$ with $\overline{C}(\bar{x}) = 0$. Let $\vecc{s}_X$ be the strategy profile for the input players of $\mathcal{G}^d_\mu$ corresponding to $\bar{x}$. By \Cref{lem:circuit_NANDsemantics} this profile can be extended to an $\alpha$-PNE,
  where $G_1$ is playing her zero strategy.
  As the gate players follow the NAND semantics in this PNE, the cost of player $G_1$ is exactly $\costZero{1}(1) = \lambda \mu$.

  If, on the other hand, for all bit vectors $x \in \set{0,1}^n$ we have $C'(x) = 0$, then for all $\bar{x} \in \set{0,1}^n$ we have $\overline{C}(\bar{x}) = 1$. Thus, using \Cref{lem:circuit_NANDsemantics} we know that in every $\alpha$-PNE $G_1$ plays her one strategy.
  As $G_1$ follows the NAND semantics in any $\alpha$-PNE and the player corresponding to one of the inputs of $g_1$ is the static player, we obtain that the cost of $G_1$ is exactly $\costOne{1}(2) = \mu 2^d$.
  Noticing that $\lambda = 3^{\nicefrac{d}{2}} < 2^d$, we have deduced the following: either $C'$ is a YES-instance, and $\mathcal{G}^d_\mu$ has an $\alpha$-PNE where $G_1$ has a cost of (at most) $\lambda \mu$; or $C'$ is a NO-instance, and for every $\alpha$-PNE of $\mathcal{G}^d_\mu$, $G_1$ has a cost of (at least) $2^d \mu$. This immediately implies that determining whether an $\alpha$-PNE exists in which a certain player has cost at most $z$ is NP-hard for $\lambda\mu<z<2^d\mu$.
  To prove that the problem remains NP-hard for an arbitrary $z>0$, simply take a rational $c$ such that $c\lambda\mu<z<c2^d\mu$ and rescale all costs of the resources in $\mathcal{G}^d_\mu$ by $c$.

\end{proof}

\section{Hardness of Existence} \label{sec:hardness_existence}

In this section we show that it is NP-hard to decide whether a polynomial congestion
game has an $\alpha$-PNE. For this we use a black-box reduction: our hard instance is obtained by
combining any
(weighted) polynomial congestion game $\mathcal{G}$ without $\alpha$-PNE (i.e., the
game from \Cref{sec:nonexistence}) with the circuit gadget of the previous section.
To achieve this, it would be convenient to make some assumptions on the game $\mathcal{G}$, which however do not influence the existence or nonexistence of approximate equilibria.

\paragraph*{Structural Properties of $\mathcal{G}$}
Without loss of generality, we assume that a weighted polynomial congestion game of degree $d$ has the following structural properties.

\begin{itemize}
  \item \emph{No player has an empty strategy.} If, for some player $i$, $\emptyset\in S_i$, then this strategy would be $\alpha$-dominating for $i$. Removing $i$ from the game description would not affect the (non)existence of (approximate) equilibria\footnote{By this we mean, if $\mathcal{G}$ has (resp.\ does not have) $\alpha$-PNE, then $\tilde{\mathcal{G}}$, obtained by removing player $i$ from the game, still has (resp.\ still does not have) $\alpha$-PNE.}.

  \item \emph{No player has zero weight.} If a player $i$ had zero weight, her
  strategy would not influence the costs of the strategies of the other players.
  Again, removing $i$ from the game description would not affect the (non)existence
  of equilibria.

  \item \emph{Each resource $e$ has a monomial cost function with a strictly positive coefficient}, i.e.\ $c_e(x)=a_e x^{k_e}$ where $a_e>0$ and $k_e\in\{0,\ldots,d\}$. If a resource had a more general cost function $c_e(x)=a_{e,0}+a_{e,1}x+\ldots+ a_{e,d} x^{d}$, we could split it into at most $d+1$ resources with (positive) monomial costs, $c_{e,0}(x)=a_{e,0}$, $c_{e,1}(x)=a_{e,1} x$, \ldots , $c_{e,d}(x)=a_{e,d} x^d$.
  These monomial cost resources replace the original resource, appearing on every strategy that included $e$.

  \item \emph{No resource $e$ has a constant cost function.} If a resource $e$ had a constant cost function $c_e(x)=a_{e,0}$, we could replace it by new resources having monomial cost.
   For each player $i$ of weight $w_i$, replace resource $e$ by a resource $e_i$ with monomial cost $c_{e_i}(x)=\frac{a_{e,0}}{w_i}x$, that is used exclusively by player $i$ on her strategies that originally had resource $e$.
   Note that $c_{e_i}(w_i) = a_{e,0}$, so that this modification does not change the player's costs, neither has an effect on the (non)existence of approximate equilibria.
   If a resource has cost function constantly equal to zero, we can simply remove it from the description of the game.
\end{itemize}
For a game having the above properties, we define the (strictly positive) quantities
\begin{equation}
  \label{eq:bb_hardness_aux}
  a_{\min} = \min_{e\in E} a_e,\quad W = \sum_{i\in N}w_i,\quad c_{\max} = \sum_{e\in E}c_e(W).
\end{equation}
Note that $c_{\max}$ is an upper bound on the cost of any player on any strategy profile.

\paragraph*{Rescaling of $\mathcal{G}$}
In our construction of the combined game we have to make sure that the weights of the players in $\mathcal{G}$ are smaller than the weights of the players in the circuit gadget. We introduce the following rescaling argument.

For any $\gamma\in(0,1]$ define the game $\tilde{\mathcal{G}}_\gamma$, where we rescale the player weights and resource cost coefficients in $\mathcal{G}$ as
\begin{equation}
  \label{eq:bb_hardness_scaling}
  \tilde{a}_e=\gamma^{d+1-k_e} a_e,\quad \tilde{w}_i=\gamma w_i,\quad \tilde{c}_e(x)=\tilde{a}_e x^{k_e}.
\end{equation}
This changes the quantities in \eqref{eq:bb_hardness_aux} for $\tilde{\mathcal{G}}_\gamma$ to (recall that $k_e\geq 1$)
\begin{align*}
\tilde{a}_{\min}&=\min_{e\in E} \tilde{a}_e=\min_{e\in E}\gamma^{d+1-k_e} a_e\geq\gamma^{d}\min_{e\in E}a_e=\gamma^d a_{\min},\\
\tilde{W}&=\sum_{i\in N}\tilde{w}_i=\sum_{i\in N}\gamma w_i=\gamma W,\\
\tilde{c}_{\max}&=\sum_{e\in E}\tilde{c}_e(\tilde{W})=\sum_{e\in E}\tilde{a}_e(\gamma W)^{k_e}= \sum_{e\in E}\gamma^{d+1}a_e W^{k_e}=\gamma^{d+1}\sum_{e\in E}c_e(W)=\gamma^{d+1}c_{\max}.
\end{align*}
In $\tilde{\mathcal{G}}_\gamma$ the player costs are all uniformly scaled as $\tilde{C}_i(\vecc{s})=\gamma^{d+1} C_i(\vecc{s})$, so that the Nash dynamics and the (non)existence of equilibria are preserved.

The next lemma formalizes the combination of both game gadgets and, furthermore, establishes the gap-introduction in the equilibrium factor. Using it, we will derive our key hardness tool of~\cref{th:hardness_existence}.
\begin{lemma}\label{th:bb_hardness}
Fix any integer $d\geq 2$ and real $\alpha\geq 1$. Suppose there exists a weighted polynomial congestion game $\mathcal{G}$ of degree $d$ that does \emph{not} have an $\alpha$-PNE. Then, for any circuit $C$ there exists a game $\tilde{\mathcal{G}}_C$ with the following property: the sets of $\alpha$-PNE and exact PNE of $\tilde{\mathcal{G}}_C$ coincide and are in one-to-one correspondence with the set of satisfying assignments of $C$. In particular, one of the following holds: either
\begin{enumerate}
\item $C$ has a satisfying assignment, in which case $\tilde{\mathcal{G}}_C$ has an exact PNE (and thus, also an $\alpha$-PNE); or

\item $C$ has no satisfying assignments, in which case $\tilde{\mathcal{G}}_C$ has no $\alpha$-PNE (and thus, also no exact PNE).
\end{enumerate}
\end{lemma}

\begin{proof}
Let $\mathcal{G}$ be a congestion game as in the statement of the theorem having the above mentioned structural properties.
Recalling that weighted polynomial congestion games of degree $d$ have $d$-PNE \cite{Caragiannis:2019aa}, this implies that $\alpha<d<3^{\nicefrac{d}{2}}$. Fix some $0<\varepsilon<3^{\nicefrac{d}{2}}-\alpha$ and take $\mu\geq 1+\frac{3^{d + \nicefrac{d}{2}}}{ \min \set{\varepsilon,1}}$; in this way $\alpha<3^{\nicefrac{d}{2}}-\varepsilon\leq3^{\nicefrac{d}{2}}-\varepsilon(\mu)$.

Given a circuit $C$ we construct the game $\tilde{\mathcal{G}}_C$ as follows.
We combine the game $\mathcal{G}^d_\mu$ whose Nash dynamics model the
NAND semantics of $C$, as described in \cref{sec:circuit}, with the game $\tilde{\mathcal{G}}_\gamma$
obtained from $\mathcal{G}$ via the aforementioned rescaling.
We choose $\gamma \in (0,1]$ sufficiently small such that the following three inequalities hold for the quantities in \eqref{eq:bb_hardness_aux} for $\mathcal{G}$:
\begin{equation}
  \label{eq:bb_hardness_gamma}
  \gamma W<1,
  \quad \gamma\sum_{e\in E}a_e<\frac{\mu}{\mu-1}\left(\frac{3}{2}\right)^d,
  \quad \gamma\alpha^2<\frac{a_{\min}}{c_{\max}}.
\end{equation}
The set of players in $\tilde{\mathcal{G}}_C$ corresponds to the (disjoint) union of the
static, input and gate players in $\mathcal{G}^d_\mu$ (which all have weights $1$) and the players in
$\tilde{\mathcal{G}}_\gamma$ (with weights $\tilde w_i$). We also consider a new dummy resource with constant cost
$c_{\mathrm{dummy}}(x)=\frac{\tilde{a}_{\min}}{\alpha}$. Thus, the set of resources
corresponds to the (disjoint) union of the gate resources $0_k,1_k$ in $\mathcal{G}^d_\mu$,
the resources in $\tilde{\mathcal{G}}_\gamma$, and the dummy resource. We augment the strategy space
of the players as follows:
\begin{itemize}
\item each input player or gate player of $\mathcal{G}^d_\mu$ that is \emph{not} the output player $G_1$ has the same strategies as in $\mathcal{G}^d_\mu$ (i.e.\ either the zero or the one strategy);
\item the zero strategy of the output player $G_1$ is the same as in $\mathcal{G}^d_\mu$, but her one strategy is augmented with \emph{every} resource in $\tilde{\mathcal{G}}_\gamma$; that is, $s^1_{G_1}=\{1_1\}\cup E(\tilde{\mathcal{G}}_\gamma)$;
\item each player $i$ in $\tilde{\mathcal{G}}_\gamma$ keeps her original strategies as in $\tilde{\mathcal{G}}_\gamma$, and gets a new dummy strategy $s_{i,\mathrm{dummy}}=\{\mathrm{dummy}\}$.
\end{itemize}
A graphical representation of the game $\tilde{\mathcal{G}}_C$ can be seen
in~\cref{fig:mergegame}.

\begin{figure}[t]
  \centering
   \includegraphics[trim={0 23.5cm 12cm 0}, clip]{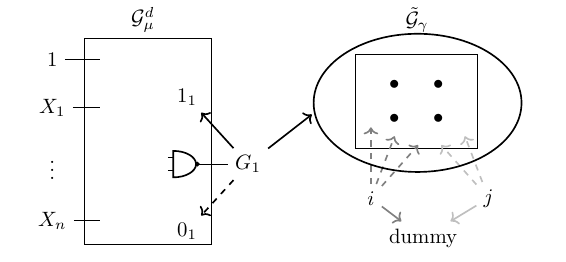}
  \caption{Merging a circuit game (on the left) and a game without approximate equilibria (on the right). Changes to the subgames are indicated by solid arrows.
  The new one strategy of $G_1$ consists of $1_1$ and all resources in $\tilde{\mathcal{G}}_\gamma$, while the zero strategy stays unchanged. The players of $\tilde{\mathcal{G}}_\gamma$ get a new strategy (the dummy resource), and keep their old strategies playing in $\tilde{\mathcal{G}}_\gamma$ .
  }
  \label{fig:mergegame}
\end{figure}

To finish the proof, we need to show that every $\alpha$-PNE of $\tilde{\mathcal{G}}_C$ is an exact PNE and corresponds to a satisfying assignment of $C$; and, conversely, that every satisfying assignment of $C$ gives rise to an exact PNE of $\tilde{\mathcal{G}}_C$ (and thus, an $\alpha$-PNE as well).

Suppose that $\vecc{s}$ is an $\alpha$-PNE of $\tilde{\mathcal{G}}_C$, and let $\vecc{s}_X$ denote the strategy profile restricted to the input players of $\mathcal{G}^d_\mu$. Then, as in the proof of \cref{lem:circuit_NANDsemantics}, every gate player that is not the output player must respect the NAND semantics, and this is an $\alpha$-dominating strategy. For the output player, either $\vecc{s}_X$ is a non-satisfying assignment, in which case the zero strategy of $G_1$ was $\alpha$-dominating,
and this remains $\alpha$-dominating in the game $\tilde{\mathcal{G}}_C$ (since only the cost of the one strategy increased for the output player); or $\vecc{s}_X$ is a satisfying assignment.
In the second case, we now argue that the one strategy of $G_1$ remains $\alpha$-dominating. The cost of the output player on the zero strategy is at least $c_{0_1}(2)=\lambda\mu2^d$, and the cost on the one strategy is at most
\[ c_{1_1}(2)+\sum_{e\in E}\tilde{c}_e(1+\gamma W)=\mu 2^d+\sum_{e\in E}\gamma^{d+1-k_e}a_e(1+\gamma W)^{k_e}<\mu 2^d+\gamma\sum_{e\in E}a_e 2^d<\mu 2^d+\frac{\mu}{\mu-1}3^d,
\]
where we used the first and second bounds from \eqref{eq:bb_hardness_gamma}. Thus, the ratio between the costs is at least
\[\frac{\lambda\mu2^d}{\mu 2^d+\frac{\mu}{\mu-1}3^d}=\lambda\left(\frac{1}{1+\frac{1}{\mu-1}\left(\frac{3}{2}\right)^d}\right)
    > 3^{\nicefrac{d}{2}} \parens{\frac{1}{1 + \frac{1}{\mu - 1} 3^d } }
    > 3^{\nicefrac{d}{2}} - \varepsilon(\mu)
    >\alpha.
\]

Given that the gate players must follow the NAND semantics, the input players are also locked to their strategies (i.e.\ they have no incentive to change) due to the proof of \cref{lem:circuit_NANDsemantics}. The only players left to consider are the players from $\tilde{\mathcal{G}}_\gamma$.
First we show that, since $\vecc{s}$ is an $\alpha$-PNE, the output player must be playing her one strategy.
If this was not the case, then each dummy strategy of a player in $\tilde{\mathcal{G}}_\gamma$ is $\alpha$-dominated by any other strategy: the dummy strategy incurs a cost of $\frac{\tilde{a}_{\min}}{\alpha}\geq\gamma^d\frac{a_{\min}}{\alpha}$, whereas any other strategy would give a cost of at most $\tilde{c}_{\max}=\gamma^{d+1}c_{\max}$ (this is because the output player is not playing any of the resources in $\tilde{\mathcal{G}}_\gamma$). The ratio between the costs is thus at least
\[ \frac{\gamma^da_{\min}}{\gamma^{d+1}c_{\max}\alpha}=\frac{a_{\min}}{\gamma c_{\max}\alpha}>\alpha.
\]
Since the dummy strategies are $\alpha$-dominated, the players in $\tilde{\mathcal{G}}_\gamma$ must be playing on their original sets of strategies. The only way for $\vecc{s}$ to be an $\alpha$-PNE would be if $\mathcal{G}$ had an $\alpha$-PNE to begin with, which yields a contradiction.
Thus, the output player is playing the one strategy (and hence, is present in every resource in $\tilde{\mathcal{G}}_\gamma$). In such a case, we can conclude that each dummy strategy is now $\alpha$-dominating.
If a player $i$ in $\tilde{\mathcal{G}}_\gamma$ is not playing a dummy strategy, she is playing at least one resource in $\tilde{\mathcal{G}}_\gamma$, say resource $e$. Her cost is at least  $\tilde{c}_e(1+\tilde{w}_i)=\tilde{a}_e(1+\tilde{w}_i)^{k_e}>\tilde{a}_e\geq\tilde{a}_{\min}$ (the strict inequality holds since, by the structural properties of our game, all of $\tilde{a}_e$, $\tilde{w}_i$ and $k_e$ are strictly positive quantities). On the other hand, the cost of playing the dummy strategy is $\frac{\tilde{a}_{\min}}{\alpha}$.
Thus, the ratio between the costs is greater than $\alpha$.

We have concluded that, if $\vecc{s}$ is an $\alpha$-PNE of $\tilde{\mathcal{G}}_C$, then $\vecc{s}_X$ corresponds to a satisfying assignment of $C$, all the gate players are playing according to the NAND semantics, the output player is playing the one strategy, and all players of $\tilde{\mathcal{G}}_\gamma$ are playing the dummy strategies. In this case, we also have observed that each player's current strategy is $\alpha$-dominating, so the strategy profile is an exact PNE.
To finish the proof, we need to argue that every satisfying assignment gives rise to a unique $\alpha$-PNE. Let $\vecc{s}_X$ be the strategy profile corresponding to this assignment for the input players in $\mathcal{G}^d_\mu$.
Then, as before, there is one and exactly one $\alpha$-PNE $\vecc{s}$ in $\tilde{\mathcal{G}}_C$ that agrees with $\vecc{s}_X$; namely, each gate player follows the NAND semantics, the output player plays the one strategy, and the players in $\tilde{\mathcal{G}}_\gamma$ play the dummy strategies.
\end{proof}

\begin{theorem}
\label{th:hardness_existence}
For any integer $d\geq 2$ and real $\alpha\geq 1$, suppose there exists a weighted polynomial congestion game which does not have an $\alpha$-PNE. Then it is NP-hard to decide whether (weighted) polynomial congestion games of degree $d$ have an
$\alpha$-PNE. If in addition $\alpha$ is polynomial-time computable,\footnote{Recall the definition of polynomial-time computable real number at the end of \cref{sec:model}.}
the aforementioned problem is NP-complete.
\end{theorem}

\begin{proof}
Let $d\geq 2$ and $\alpha\geq 1$. Let $\mathcal{G}$ be a weighted polynomial congestion game of degree $d$ that has no $\alpha$-PNE; this means that for every strategy profile $\vecc{s}$ there exists a player $i$ and a strategy $s'_i\neq s_i$ such that $C_i(s_i,\vecc{s}_{-i})>\alpha \cdot C_i \parens{s'_i,\vecc{s}_{-i}}$.
Note that the functions $C_i$ are polynomials of degree $d$ and hence they are continuous on the weights $w_i$ and the coefficients $a_e$ appearing on the cost functions. Hence, any arbitrarily small perturbation of the $w_i,a_e$ does not change the sign of the above inequality. Thus, without loss of generality, we can assume that all $w_i,a_e$ are rational numbers. By a similar reasoning, we can let $\bar{\alpha}>\alpha$ be a rational number sufficiently close to $\alpha$ such that $\mathcal{G}$ still does not have an $\bar{\alpha}$-PNE.

Next, we consider the game $\tilde{\mathcal{G}}_\gamma$ obtained from $\mathcal{G}$ by rescaling, as in the proof of \cref{th:bb_hardness}, but with $\bar{\alpha}$ playing the role of $\alpha$. Notice that the rescaling is done via the choice of a sufficiently small $\gamma$, according to \eqref{eq:bb_hardness_gamma}, and hence in particular we can take $\gamma$ to be a sufficiently small rational. In this way, all the player weights and coefficients in the cost of resources are rational numbers scaled by a rational number and hence rationals.

Finally, we are able to provide the desired NP reduction from \textsc{Circuit Satisfiability}. Given a Boolean circuit $C'$ built with 2-input NAND gates, transform it into a valid circuit $C$ in canonical form. From $C$ we can construct in polynomial time the game $\tilde{\mathcal{G}}_C$ as described in the proof of \cref{th:bb_hardness}. The `circuit part', i.e.\ the game $\mathcal{G}^d_\mu$, is obtained in polynomial time from $C$, as in the proof of \cref{th:hardness_PNE}; the description of the game $\tilde{\mathcal{G}}_\gamma$ involves only rational numbers,
and hence the game can be represented by a constant number of bits (i.e.\ independent of the circuit $C$). Similarly, the additional dummy strategy has a constant delay of $\nicefrac{\tilde{a}_{\min}}{\bar{\alpha}}$,
and can be represented with a single rational number. Merging both $\mathcal{G}^d_\mu$ and $\tilde{\mathcal{G}}_\gamma$ into a single game $\tilde{\mathcal{G}}_C$ can be done in linear time. Since $C$ has a satisfying assignment iff $\tilde{\mathcal{G}}_C$ has an $\alpha$-PNE (or $\bar{\alpha}$-PNE), this concludes that the problem described is NP-hard.

If $\alpha$ is polynomial-time computable, the problem is clearly in NP: given a weighted polynomial congestion game of degree $d$ and a strategy profile $\vecc{s}$, one can check if $\vecc{s}$ is an $\alpha$-PNE by computing the ratios between the cost of each player in $\vecc{s}$ and their cost for each possible deviation, and comparing these ratios with $\alpha$.
\end{proof}

Combining the hardness result of \cref{th:hardness_existence} together with the
nonexistence result of \cref{th:nonexistence} we get the following corollary, which
is the main result of this section.

\begin{corollary}
\label{th:hardness_existence2}
For any integer $d\geq 2$ and real $\alpha\in[1,\alpha(d))$, it is NP-hard to decide whether (weighted) polynomial congestion games of degree $d$ have an
$\alpha$-PNE, where $\alpha(d)=\tilde{\varOmega}(\sqrt{d})$ is the same
as in~\cref{th:nonexistence}. If in addition $\alpha$ is polynomial-time computable, the aforementioned problem is NP-complete.
\end{corollary}

Notice that, in the proof of \cref{th:bb_hardness,th:hardness_existence}, we constructed a polynomial-time reduction from \textsc{Circuit Satisfiability} to the problem of determining whether a given congestion game has an $\alpha$-PNE. Not only does this reduction map YES-instances of one problem to YES-instances of the other, but it also induces a bijection between the sets of satisfying assignments of a circuit $C$ and $\alpha$-PNE of the corresponding game $\tilde{\mathcal{G}}_C$. That is, this reduction is \emph{parsimonious}. As a consequence, we can directly lift hardness of problems associated with counting satisfying assignments to \textsc{Circuit Satisfiability} into problems associated with counting $\alpha$-PNE of congestion games:

\begin{corollary}
\label{th:hardness_counting}
Let $k\geq 1$ and $d\geq 2$ be integers and $\alpha\in[1,\alpha(d))$ where $\alpha(d)=\tilde{\varOmega}(\sqrt{d})$ is the same as in~\cref{th:nonexistence}. Then

\begin{itemize}
\item it is {\#}P-hard to count the number of $\alpha$-PNE of (weighted) polynomial congestion games of degree $d$;
\item it is NP-hard to decide whether a (weighted) polynomial congestion game of degree $d$ has at least $k$ distinct $\alpha$-PNE.
\end{itemize}
\end{corollary}

\begin{proof}
The hardness of the first problem comes from the {\#}P-hardness of the counting version of \textsc{Circuit Satisfiability} (see, e.g., \citep[Ch.~18]{PapadimitriouComplexity}). For the hardness of the second problem, it is immediate to see that the following problem is NP-complete, for any fixed integer $k\geq 1$: given a circuit $C$, decide whether there are at least $k$ distinct satisfying assignments for $C$ (simply add ``dummy'' variables to the description of the circuit).
\end{proof}

\section{General Cost Functions} \label{sec:general_costs}

In this final section we leave the domain of polynomial latencies and study the
existence of approximate equilibria in general congestion games having arbitrary
(nondecreasing) cost functions. Our parameter of interest, with respect to which
both our positive and negative results are going to be stated, is the number
of players $n$. We start by showing that $n$-PNE always exist:

\begin{theorem}
\label{th:existence_general_n}
Every weighted congestion game with $n$ players and arbitrary (nondecreasing) cost functions has an $n$-PNE.
\end{theorem}

\begin{proof}
Fix a weighted congestion game with $n\geq 2$ players, some strategy profile $\vecc{s}$, and a possible deviation $s'_i$ of player $i$. First notice that we can bound the change in the cost of any other player $j\neq i$ as

\begin{align}
C_{j}(s'_i,\vecc s_{-i}) - C_{j}(\vecc s)
&= \sum_{e\in s_j} c_e\parens{x_e(s'_i,\vecc s_{-i})} - \sum_{e\in s_j} c_e \parens{x_e(\vecc s)} \nonumber\\
&\!\begin{multlined}[b]
     = \sum_{e\in s_j\inters(s'_i\setminus s_i)} \left[c_e \parens{x_e(s'_i,\vecc s_{-i})} - c_e \parens{x_e(\vecc s)}\right] \\
          +\sum_{e\in s_j\inters(s_i\setminus s'_i)} \left[c_e \parens{x_e(s'_i,\vecc s_{-i})} - c_e \parens{x_e(\vecc s)}\right]
     \end{multlined} \label{eq:helper_3}\\
& \leq \sum_{e\in s_j\inters(s'_i\setminus s_i)} \left[c_e \parens{x_e(s'_i,\vecc s_{-i})} - c_e \parens{x_e(\vecc s)}\right]\nonumber \\
&\leq \sum_{e\in s'_i} c_e \parens{x_e(s'_i,\vecc s_{-i})}\nonumber \\
&= C_i(s'_i,\vecc s_{-i}), \label{eq:helper_2}
\end{align}
the first inequality holding due to the fact that the second sum in~\eqref{eq:helper_3} contains only nonpositive terms (since the latency functions are nondecreasing).

Next, define the social cost $C(\vecc{s})=\sum_{i\in N}C_i(\vecc{s})$. Adding the above inequality over all players $j\neq i$ (of which there are $n-1$) and rearranging, we successively derive:

\begin{align}
  \sum_{j\neq i}C_{j}(s'_i,\vecc s_{-i}) - \sum_{j\neq i}C_{j}(\vecc s) &\leq (n-1)C_i(s'_i,\vecc s_{-i}) \nonumber\\
  \left(C(s'_i,\vecc s_{-i})-C_i(s'_i,\vecc s_{-i})\right)-\left(C(\vecc{s})-C_i(\vecc{s})\right) &\leq (n-1)C_i(s'_i,\vecc s_{-i}) \nonumber\\
  \qquad C(s'_i,\vecc s_{-i})-C(\vecc{s}) &\leq nC_i(s'_i,\vecc s_{-i})-C_i(\vecc{s}).\label{eq:napxpotential}
\end{align}
We conclude that, if $s'_i$ is an $n$-improving deviation for player $i$ (i.e., $nC_i(s'_i,\vecc s_{-i})<C_i(\vecc{s})$), then the social cost must strictly decrease after this move. Thus, any (global or local) minimizer of the social cost must be an $n$-PNE (the existence of such a minimizer is guaranteed by the fact that the strategy spaces are finite).
\end{proof}

The above proof not only establishes the existence of $n$-PNE in
general congestion games, but also highlights a few additional interesting features.
First, we mention that for polynomial cost functions of degree at most $d$, it was shown in \cite{Caragiannis:2019aa},\cite{CARAGIANNIS2021} that the (weighted) social cost decreases at every $(d+1)$-improving move; here we have shown, in a similar spirit, that for general cost functions the (unweighted) social cost decreases at every $n$-improving move.
Second, due to the key inequality~\eqref{eq:napxpotential}, $n$-PNE are reachable via
sequences of $n$-improving moves, in addition to arising also as minimizers of the
social cost function. These attributes give a nice ``constructive'' flavour
to~\cref{th:existence_general_n}.
Third, exactly because social cost optima are $n$-PNE, the \emph{Price of
Stability}\footnote{The Price of Stability (PoS) is a well-established and
extensively studied notion in algorithmic game theory, originally studied
in~\citep{ADKTWR04,Correa2004}. It captures the minimum approximation ratio of the
social cost between equilibria and the optimal solution (see, e.g.,
\citep{Christodoulou2015,cggs2018-journal}); in other words, it is the best-case
analogue of the the Price of Anarchy (PoA) notion of~\citet{Koutsoupias2009a}.} of
$n$-PNE is optimal (i.e., equal to $1$) as well.
Another, more succinct way, to interpret these observations is within the context of
\emph{approximate potentials} (see, e.g.,
\citep{Chen2008,Christodoulou2011a,cggs2018-journal}); \eqref{eq:napxpotential}
establishes that the social cost itself is always an $n$-approximate potential of
any congestion game.

Next, we design a family of games that do not admit $\varTheta\left(\frac{n}{\ln
n}\right)$-PNE, thus nearly matching the upper bound~\cref{th:existence_general_n}.

\begin{theorem}
\label{th:nonexistence_general_costs}
For any integer $n\geq 2$, there exist weighted congestion games with $n$ players and general (nondecreasing) cost functions that do not have $\alpha$-PNE for any $\alpha<\Phi_{n-1}$, where
$\Phi_m\sim\frac{m}{\ln m}$ is the unique positive solution of $(x+1)^m=x^{m+1}$.
\end{theorem}

\begin{proof}
For any integer $n\geq 2$, let
$\xi=\Phi_{n-1}$ be the positive solution of $(x+1)^{n-1}=x^{n}$. Then,
equivalently,
\begin{equation}
\label{eq:helper_5}
\left(1+\frac{1}{\xi}\right)^{n-1}=\xi.
\end{equation}
Furthermore, as we
mentioned in~\cref{sec:model}, $\xi>1$ and asymptotically $\Phi_{n-1}\sim\frac{n}{\ln n}$.

Consider the following congestion game $\mathcal{G}_n$.
There are $n=m+1$ players $0,1, \ldots, m$, where player $i$ has weight $w_i=\nicefrac{1}{2^i}$. In particular, this means that for any $i \in \set{1,\dots,m}$:
\begin{equation}
  \sum_{k=i}^m w_k< w_{i-1}\leq w_0.
\label{eq:weights}
\end{equation}

Furthermore, there are $2(m+1)$ resources $a_0, a_1, \ldots, a_m, b_0, b_1, \ldots, b_m$, where resources $a_i$ and $b_i$ have the same cost function $c_i$ given by
\[
  c_{a_0}(x)=c_{b_0}(x)=c_0(x)=
  \begin{cases}
     1,   & \text{if }  x\geq w_0, \\
     0,  & otherwise;
  \end{cases}
\]
and for all $i \in \set{1,\dots,m}$,\\
\[
  c_{a_i}(x)=c_{b_i}(x)=c_i(x)=
  \begin{cases}
    \frac{1}{\xi}\left(1+\frac{1}{\xi}\right)^{i-1},  & \text{if }  x\geq w_0 + w_i, \\
    0,  & otherwise.
  \end{cases}
\]
The strategy set of player $0$ and of all players $i \in \set{1,\dots,m}$ are, respectively,
\[
  S_0=\{\{a_0, \ldots, a_m\},\{b_0, \ldots, b_m\}\},
  \qquad \text{and} \qquad
  S_i=\{\{a_0, \ldots, a_{i-1}, b_i\},\{b_0, \ldots, b_{i-1}, a_i\}\}.
\]
We show that this game has no $\alpha$-PNE, for any $\alpha<\xi$, by proving that in any outcome there is at least one player that can
deviate and improve her cost by a factor of at least $\xi$. Due to symmetry it is
sufficient to consider the following two kinds of outcomes:

\emph{\underline{Case 1:} Player $0$ is alone on resource $a_0$.}

Then player $0$ must have chosen $\{a_0, \ldots, a_m\}$, and all other players
$i \in \set{1,\dots,m}$ must have chosen strategy $\{b_0, \ldots, b_{i-1}, a_i\}$. In this
outcome, player $0$ has a cost of
\begin{align*}
  c_0(w_0)+ \sum_{i=1}^{m} c_i(w_0+w_i)
  = 1 + \frac{1}{\xi}  \sum_{i=1}^{m}\left(1+\frac{1}{\xi}\right)^{i-1}
  =\left(1+\frac{1}{\xi}\right)^{m}
  = \xi,
\end{align*}
where the last equality follows by the fact that $m=n-1$ and~\eqref{eq:helper_5}.
Deviating to $\{b_0, \ldots ,b_m\}$, player 0 would get a cost of
\[
   c_0(w_0+\ldots+ w_m)+ \sum_{i=1}^{m} c_i\left(w_0+\sum_{j=i+1}^m w_j\right) = 1+0,
\]
where we used $w_0+\sum_{j=i+1}^m w_j < w_0 + w_i$ (see \eqref{eq:weights}).

Thus player $0$ can improve by a factor of at least
$\xi$.

\emph{\underline{Case 2:} Player 0 is sharing resource $a_0$ with at least one other player.}

Let $j$ be the smallest index of such a player, i.e., player $j$ plays $\{a_0,
\ldots, a_{j-1}, b_j\}$ and all players $i \in \set{1,\dots,j-1}$ have chosen strategy $\{b_0,
\ldots, b_{i-1}, a_i\}$. In such a profile the cost of player $j$ is at least
\begin{align*}
  c_0(w_0+w_j)+ \sum_{i=1}^{j-1} c_i(w_0+w_i+w_j)
  = 1 + \frac{1}{\xi} \sum_{i=1}^{j-1} \left(1+\frac{1}{\xi}\right)^{i-1}
  =\left(1+\frac{1}{\xi}\right)^{j-1},
\end{align*}
while deviating to $j$'s other strategy would result in a cost of at most
\begin{align*}
  c_0\left(\sum_{i=1}^{m} w_i\right) +  \sum_{i=1}^{j-1} c_i\left(\sum_{k=i+1}^m w_k\right) + c_j\left(w_0+\sum_{k=j}^m w_k\right)
  = 0 + 0 +  \frac{1}{\xi}\left(1+\frac{1}{\xi}\right)^{j-1}.
\end{align*}
For the last equality we used $\sum_{i=1}^{m} w_i < w_0 $ and $\sum_{k=i+1}^m w_k < w_i $ from \eqref{eq:weights}.

Thus player $j$ can improve by a factor of at least
$\xi$.
\end{proof}

Similar to the spirit of the rest of our paper so far, we'd like to show an
NP-hardness result for deciding existence of $\alpha$-PNE for general games as well.
We do exactly that in the following theorem, where now $\alpha$ grows as
$\tilde{\varTheta}(n)$. Again, we use the circuit gadget and combine it with the game
from the previous nonexistence~\cref{th:nonexistence_general_costs}.
The main difference to the previous reductions is that our approximation bound $\alpha$ is not fixed, but depends on the number of players $n$.
On the
other hand we are not restricted to polynomial latencies, so we use step functions
having a single breakpoint.
 We want to emphasize that in the following theorem, $\varepsilon$ is fixed (i.e., not part of the input).

\begin{theorem}
\label{th:hardness_existence_general_alt} Let $\varepsilon>0$, and let $\tilde{\alpha}:\nbb_{\geq2} \map \R$ be any
sequence of reals such that $1\leq \tilde\alpha(n) <
\frac{\Phi_{n-1}}{1+\varepsilon}=\tilde{\varTheta}(n)$,
where $\Phi_m\sim\frac{m}{\ln m}$ is the unique positive solution of $(x+1)^m=x^{m+1}$.
Then, it is NP-hard to decide whether a (weighted) congestion game $\mathcal{G}$ has an $\tilde{\alpha}(n_{\mathcal{G}})$-PNE (where $n_{\mathcal{G}}$ is the number of players of $\mathcal{G}$).
 If in addition $\tilde{\alpha}$ is a polynomial-time computable real sequence (as defined in \cref{sec:model}), the aforementioned problem is NP-complete.
\end{theorem}

\begin{proof}
Recall that we have $\Phi_{n-1}\sim\frac{n}{\ln n}$. Given $\varepsilon>0$, without loss of generality assume $\varepsilon<1$, so that $1+\nicefrac{\varepsilon}{3}<(1+\varepsilon)(1-\nicefrac{\varepsilon}{3})$. Let $n_0,\ell$ be large enough natural numbers such that

\begin{equation}
  1+\frac{1}{\ell}<\frac{(1+\varepsilon)(1-\frac{\varepsilon}{3})}{1+\frac{\varepsilon}{3}}
  \qquad\text{and}\qquad
  \left(1-\frac{\varepsilon}{3}\right)\frac{n}{\ln n}\leq\Phi_{n-1}\leq\left(1+\frac{\varepsilon}{3}\right)\frac{n}{\ln n}
  \quad\text{for all}\;\;n\geq n_0.
  \label{eq:hardness_existence_general_aux1}
\end{equation}

We will again reduce from \textsc{Circuit Satisfiability}: given a circuit $C$, we must construct (in polynomial time) a game $\tilde{\mathcal{G}}$, say with $\tilde{n}$ players, that has an $\tilde{\alpha}(\tilde{n})$-PNE if and only if $C$ has a satisfying assignment. Without loss of generality assume that $C$ is in canonical form (as described in \cref{sec:circuit}); add also one extra gate that negates the output of $C$, making this the new output of a circuit $\bar{C}$, say with $m$ inputs and $K$ NAND gates.
Let $s= m+K+1$, $n= \ell s$, and take a large enough integer $d$ such that $3^{\nicefrac{d}{2}}>\Phi_{n-1}$. Note that $s$, $n$ and a suitable $d$ can all be found in time polynomial in the description of $C$. To conclude the preliminaries of this proof, assume also without loss of generality that $s\geq n_0$; if $s$ is bounded by a constant, determining whether $C$ has a satisfying assignment can be done in constant time.

Next, given $\bar{C}$ and $d$, construct the game $\mathcal{G}^d_\mu$ where $\mu$ is such that $3^{\nicefrac{d}{2}}-\varepsilon(\mu)>\Phi_{n-1}$, as in \cref{sec:circuit}. Notice that $\mathcal{G}^d_\mu$ can be computed in polynomial time from $C$, and that the $\Phi_{n-1}$-improving Nash dynamics of this game emulate the computation of the circuit. Consider also the game $\mathcal{G}_n$ with $n$ players from \cref{th:nonexistence_general_costs} that does not have $\alpha$-PNE for any $\alpha<\Phi_{n-1}$.

We would like to merge $\mathcal{G}^d_\mu$ and $\mathcal{G}_n$ into a single game $\tilde{\mathcal{G}}$, in such a way that $\tilde{\mathcal{G}}$ has an approximate PNE if and only if $C$ has a satisfying assignment. Following the same technique as in \cref{th:bb_hardness}, we would like to extend the strategies of the output player of $\mathcal{G}^d_\mu$ to include resources that are used by players in $\mathcal{G}_n$. For this technique to work, we must rescale the weights and cost functions in $\mathcal{G}_n$.
In particular, we divide all weights of the players in $\mathcal{G}_n$ by 2 (so that the sum of the weights of all the players is less than 1) and halve the breakpoints of the cost functions accordingly. We also add a new dummy resource with cost function
\[ c_{\text{dummy}}(x)=
  \begin{cases}
     \Phi_{n-1}^2,  & \text{if}\;\;  x\geq 1, \\
     0,  & \text{otherwise};
 \end{cases}
\]

We are now ready to describe the congestion game $\tilde{\mathcal{G}}$ that is obtained by merging the circuit game $\mathcal{G}^d_\mu$ with the (rescaling of) game $\mathcal{G}_n$. Note that this game has $n+s=(\ell+1) s$ players: $s$ from the circuit game (which all have weight 1) and $n$ from the nonexistence gadget. The set of resources corresponds to the union of the gate resources of $\mathcal{G}^d_\mu$, the resources in $\mathcal{G}_n$, and the dummy resource. Similarly to the proof of \cref{th:bb_hardness},

\begin{itemize}
\item we do not change the strategies of the players in $\mathcal{G}^d_\mu$, with the exception of the output player $G_1$;
\item the zero strategy of the output player $G_1$ remains the same as in $\mathcal{G}^d_\mu$, but her one strategy is augmented with the dummy resource; that is, $s^1_{G_1}=\{1_1,\text{dummy}\}$;
\item each player $i$ in $\mathcal{G}_n$ keeps her original strategies, and gets a new dummy strategy $s_{i,\text{dummy}}=\{\text{dummy}\}$.
\end{itemize}

With the above description,\footnote{This almost concludes the description of the game -- the only problem is that some of the cost functions of the game are defined in terms of $\Phi_{n-1}$, which is not a rational number. To make the proof formally correct, one can approximate $\Phi_{n-1}$ sufficiently close by a rational $\bar{\Phi}_{n-1}<\Phi_{n-1}$.
More details can be found in \Cref{appendix:rationalize}.}
the only thing left to prove NP-hardness is that $C$ has a satisfying assignment if and only if $\tilde{\mathcal{G}}$ has an $\tilde{\alpha}(n+s)$-PNE.
The proof follows the same approach as in \cref{th:bb_hardness}. Letting $\alpha<\Phi_{n-1}$, we suppose that $\tilde{\mathcal{G}}$ has an $\alpha$-PNE, say $\vecc{s}$, and proceed to prove that $C$ has a satisfying assignment.

As before, if $\vecc{s}$ is an $\alpha$-PNE, then every gate player that is not the output player must respect the NAND semantics, and this strategy is $\alpha$-dominating.
For the output player, the cost of her zero strategy remains the same, and the cost of her one strategy increases by exactly $\Phi_{n-1}^2<3^d<\frac{\mu}{\mu-1}3^d$.
Hence, if $\vecc{s}_X$ is a satisfying assignment, then the zero strategy of the output player (which negates the output of the original circuit $C$) remains $\alpha$-dominating; on the other hand, if $\vecc{s}_X$ is not a satisfying assignment, then the ratio between the costs of the zero strategy and the one strategy of the output player is at least
$$\frac{c_{0_1}(2)}{c_{1_1}(2)+\Phi_{n-1}^2}>\frac{\lambda\mu 2^d}{\mu 2^d+\frac{\mu}{\mu-1}3^d}=\lambda\left(\frac{1}{1+\frac{1}{\mu-1}\left(\frac{3}{2}\right)^d}\right)
    > 3^{\nicefrac{d}{2}} \parens{\frac{1}{1 + \frac{1}{\mu - 1} 3^d } }> 3^{\nicefrac{d}{2}} - \varepsilon(\mu)>\alpha.$$
Hence, respecting the NAND semantics remains $\alpha$-dominating for the output player as well. As a consequence, the input players are also locked to their strategies (i.e.\ they have no incentive to change).

Now, if the output player happened to be playing her one strategy, this could not be an $\alpha$-PNE. For each of the players in $\mathcal{G}_n$, the dummy strategy would incur a cost of $\Phi_{n-1}^2$, whereas any other strategy would give a cost of at most $\Phi_{n-1}$. Thus the dummy strategy would be $\Phi_{n-1}$-dominated, and the players in $\mathcal{G}_n$ must be playing on their original sets of strategies, for which we know that $\alpha$-PNE do not occur.

The above argument proves that, in an $\alpha$-PNE, the output player must be playing her zero strategy. Since the output player, by construction, negates the output of $C$, this implies that $C$ must have a satisfying assignment. This also implies that the congestion on the dummy resource cannot reach the breakpoint of 1, and hence it would be $\alpha$-dominating for each of the players in $\mathcal{G}_n$ to play her dummy strategy (and incur a cost of 0). Thus, $\vecc{s}$ is an exact PNE as well.

For the converse direction, suppose $C$ has a satisfying assignment $\vecc{s}_X$. Then this can be extended to an $\alpha$-PNE of $\tilde{\mathcal{G}}$ in which the input players play according to $\vecc{s}_X$, the gate players play according to the NAND semantics, the output player of $\mathcal{G}^d_\mu$ plays the zero strategy, and each player in $\mathcal{G}_n$ plays the dummy strategy.

We have proven that, for any $\alpha<\Phi_{n-1}$, $C$ has a satisfying assignment iff $\tilde{\mathcal{G}}$ has an $\alpha$-PNE. To conclude the proof, we verify that $\tilde{\alpha}(n+s)<\Phi_{n-1}$:

\begin{align*}
  \tilde{\alpha}(n+s)&<\frac{\Phi_{n+s-1}}{1+\varepsilon}\\
  &\leq\frac{1+\frac{\varepsilon}{3}}{1+\varepsilon}\frac{n+s}{\ln(n+s)}\\
  &\leq\frac{1+\frac{\varepsilon}{3}}{(1+\varepsilon)(1-\frac{\varepsilon}{3})}\frac{(n+s)\ln n}{n\ln(n+s)}\Phi_{n-1}\\
  &<\frac{(1+ \frac{s}{n})(1+\frac{\varepsilon}{3})}{(1+\varepsilon)(1-\frac{\varepsilon}{3})}\Phi_{n-1}\\
  &=\frac{(1+ \frac{1}{\ell})(1+\frac{\varepsilon}{3})}{(1+\varepsilon)(1-\frac{\varepsilon}{3})}\Phi_{n-1}<\Phi_{n-1}.
\end{align*}
The first inequality comes from the assumption on $\tilde{\alpha}$, the second and third come from the upper and lower bounds on $\Phi_{n}$ from \eqref{eq:hardness_existence_general_aux1} and the fact that $n+s\geq n\geq s\geq n_0$, the fourth comes from the trivial bound $\ln n<\ln (n+s)$, the equality comes from the definition of $n=\ell s$, and the final step comes from the choice of $\ell$ in \eqref{eq:hardness_existence_general_aux1}.

We conclude that the problem of deciding whether a (weighted) congestion game with $n$ players has an $\tilde{\alpha}(n)$-PNE is NP-hard. If in addition $\tilde{\alpha}$ is a polynomial-time computable real sequence, the problem is also in NP; given a game with $n$ players and a (candidate) strategy profile, verify that this is an $\tilde{\alpha}(n)$-PNE by iterating over all possible moves of all players and verifying that none of these are $\tilde{\alpha}(n)$-improving.
\end{proof}

\section{Discussion and Future Directions} \label{sec:discussion_future}

In this paper we showed that weighted congestion games with polynomial latencies of degree
$d$ do not have $\alpha$-PNE for
$\alpha<\alpha(d)=\varOmega\left(\frac{\sqrt{d}}{\ln d}\right)$. For general cost
functions, we proved that $n$-PNE always exist whereas $\alpha$-PNE in general do
not, where $n$ is the number of players and
$\alpha<\Phi_{n-1}=\varTheta\left(\frac{n}{\ln n}\right)$. We also transformed the
nonexistence results into complexity-theoretic results, establishing that deciding
whether such $\alpha$-PNE exist is itself an NP-hard problem.

We now identify two possible directions for follow-up work. A first obvious question would be
to reduce the nonexistence gap between $\varOmega\left(\frac{\sqrt{d}}{\ln
d}\right)$ (derived in~\cref{th:nonexistence} of this paper) and $d$ (shown in~\cite{Caragiannis:2019aa}) for polynomials of degree $d$; similarly for the gap between
$\varTheta\left(\frac{n}{\ln n}\right)$ (\cref{th:nonexistence_general_costs}) and $n$ (\cref{th:existence_general_n}) for general cost functions and $n$
players. Notice that all current methods for proving upper bounds (i.e.,
existence) are essentially based on potential function arguments; thus it might be
necessary to come up with novel ideas and techniques to overcome the current gaps.

A second direction would be to study the complexity of \emph{finding} $\alpha$-PNE,
when they are guaranteed to exist. For example, for polynomials of degree $d$, we
know that $d$-improving dynamics eventually reach a
$d$-PNE~\cite{Caragiannis:2019aa}, and so finding such an approximate equilibrium
lies in the complexity class PLS of local search problems (see, e.g.,
\cite{Johnson:1988aa,Schaffer:1991aa}). However, from a complexity theory
perspective the only known lower bound is the PLS-completeness of finding an
\emph{exact} equilibrium for \emph{unweighted} congestion
games~\cite{Fabrikant2004a} (and this is true even for $d=1$, i.e., affine cost
functions; see~\cite{Ackermann2008}). On the other hand, we know that $d^{O(d)}$-PNE
can be computed in polynomial time (see, e.g.,
\cite{Caragiannis2015a,gns2018_arxiv,Feldotto2017}). It would be then very
interesting to establish a ``gradation'' in complexity (e.g., from NP-hardness to
PLS-hardness to P) as the parameter $\alpha$ increases from $1$ to $d^{O(d)}$.

\bibliography{hardness_equilibria}

\appendix

\section{Technical Lemmas}

\begin{lemma}\label{lem:nonexist_ratio_2_global}
For any integer $d\geq2$ define the sequence
\[ g(d)=\left(1+d^{-\frac{1}{2(k_d+1)}}\right)^{k_d}+\frac{2\sqrt{d}}{\ln
d\left(1+\frac{\ln d}{2d}\right)^d}
\qquad\text{where}\;\; k_d=\left\lceil\frac{\ln d}{2\ln\ln d}\right\rceil .
\]
Then $\lim_{d\rightarrow\infty}g(d)=1$.
\end{lemma}

\begin{proof} Define
\[g_1(d) = \parens{1 + d^{- \frac{1}{2 (k_d + 1)}} }^{k_d}
\quad \text{ and } \quad g_2 (d) = \frac{2 \sqrt{d}}{\parens{1 + \frac{\ln d}{2
d}}^d \ln d} \] so that $g(d) = g_1(d) + g_2(d)$. We will show the desired
convergence by establishing that $\lim_{d\to\infty} g_1(d)=1$ and $\lim_{d\to\infty}
g_2(d)=0$.
We will make use of the following inequalities (see, e.g., \citep[Eq.~4.5.13]{Olver2010}):
\begin{equation}
  \exp \parens{\frac{xy}{x+y}} < \parens{1 + \frac{x}{y}}^y < \exp (x),\quad\text{for all }x,y>0.
  \label{ineq:bound_e}
\end{equation}

First, we show $\lim_{d \rightarrow \infty} g_1(d) = 1$. As $d$ and $k_d$ are
positive, we have $g_1(d) > 1$ for every $d$. Furthermore, $g_1$ is increasing in
$k_d$, and $k_d < \frac{\ln d}{2 \ln \ln d} + 1$. Thus,
\begin{equation*} g_1(d) <  \parens{1 + d ^{- \frac{1}{ \frac{\ln d}{\ln \ln d} +
  4} } }^{ \frac{\ln d}{2 \ln \ln d} + 1}.
\end{equation*}
Using the second inequality of \eqref{ineq:bound_e} with $y = \frac{\ln d}{2 \ln \ln d} + 1$ and $x = y d^{-\frac{\ln \ln d}{\ln d + 4 \ln \ln d}}$ , we can further bound
\begin{equation} g_1 (d) < \exp \parens{ \frac{\frac{\ln d}{2 \ln \ln d} + 1 }{
  d^{\frac{\ln \ln d}{\ln d + 4 \ln \ln d}} } }.
\label{eq:ub_g1}
\end{equation}
We will show that the argument of the exponential function on the r.h.s.\
of~\eqref{eq:ub_g1} goes to $0$ for $d \rightarrow \infty$, thus proving the claim.
Replacing $ \ln d = \exp \parens{\ln \ln d}$ in the numerator and $d = \exp \parens{\ln d}$ in the
denominator, that argument can be written as
\begin{equation}
  \frac{\exp \parens{\ln \ln d} \parens{\frac{1}{2 \ln \ln d} + \frac{1}{\ln d}} }{ \exp
  \parens{\frac{\ln d \ln \ln d}{\ln d + 4 \ln \ln d}} }
  =
  \parens{\frac{1}{2 \ln \ln d} + \frac{1}{\ln d}} \exp \parens{ \frac{4 (\ln \ln
  d)^2}{\ln d + \ln \ln d} } .
\label{eq:ub_g1_arg}
\end{equation}
The argument of the exponential function on the r.h.s.\ of \eqref{eq:ub_g1_arg} goes
to $0$, as $\ln d$ is the dominating term in the denominator for $d \rightarrow
\infty$. Thus, the whole expression in \eqref{eq:ub_g1_arg} goes to $0$.

Next, we show that $\lim_{d \rightarrow \infty} g_2(d) = 0$. As $d \ge 2$, we have
that $g_2(d) > 0$ for every $d$. Using the first inequality of \eqref{ineq:bound_e} with $x =
\frac{\ln d}{2}$ and $y = d$, we have
\[
  \parens{1 + \frac{\ln d}{2d}}^d > \exp \parens{\frac{d \ln d}{\ln d + 2d}}.
\] Thus, we obtain an upper bound on $g_2$ by writing $\sqrt{d} =
\exp \parens{\frac{\ln d}{2}}$:
\begin{equation} g_2(d) <
  \frac{2 \exp \parens{\frac{\ln d}{2}}}{\ln d \exp \parens{\frac{d \ln d}{\ln d + 2d}} } =
  \frac{2}{\ln d} \exp \parens{\frac{(\ln d)^2}{2 \ln d + 4d}}.
  \label{eq:g2_ub}
\end{equation} As $4d$ is the dominating term in the denominator of the argument of
the exponential function, the argument goes to $0$ for $d \rightarrow \infty$, and
thus the r.h.s. of \eqref{eq:g2_ub} goes to $0$, showing the claim.
\end{proof}

\section{\texorpdfstring{Dealing with the irrationality of $\Phi_{n-1}$ in the proof of \Cref{th:hardness_existence_general_alt}}{Dealing with the irrationality}}
\label{appendix:rationalize}

To make the proof of \Cref{th:hardness_existence_general_alt} formally correct, we choose an integer under-approximation $\bar{\Phi}_{n-1} = \floor{ \frac{n}{\ln n} }$ of $\Phi_{n-1}$ and use $\bar{\Phi}_{n-1}$ instead of $\Phi_{n-1}$ in the construction of our reduction.

First, observe that \Cref{eq:hardness_existence_general_aux1} still holds for $\bar{\Phi}_{n-1}$, i.e., for large enough $n$ we have
\begin{equation}
  \left(1-\frac{\varepsilon}{3}\right)\frac{n}{\ln n} \leq \bar{\Phi}_{n-1} < \Phi_{n-1}\leq\left(1+\frac{\varepsilon}{3}\right)\frac{n}{\ln n}.
 \label{eq:bounds-barphi}
\end{equation}

For the circuit game $\mathcal{G}^d_\mu$, we can find an integer $d$ such that $3^{\nicefrac{d}{2}} - \bar{\Phi}_{n-1} > 1$ and $3^{d/2}> \left(1+\frac{1}{\bar{\Phi}_{n-1}}\right)^{n-1}$. Note that $d$ is logarithmic in $n$ as both $\bar{\Phi}_{n-1}$ and $\left(1+\frac{1}{\bar{\Phi}_{n-1}}\right)^{n-1}$ grow asymptotically as $\frac{n}{\ln n}$, $n$ respectively.
For $\mu$ we choose an integer such that $1 \ge \frac{3^{d + \nicefrac{d}{2}}}{\mu - 1} = \varepsilon(\mu)$ (see \eqref{eq:epsilon_mu_def}); that is, $\mu \geq 1+3^{d+d/2}$.
Further, we can find a rational $\lambda$ such that
\begin{equation}
  \bar{\Phi}_{n-1}\left(1+\frac{1}{\mu-1}3^d\right)<\lambda<\frac{3^d}{\bar{\Phi}_{n-1}\left(1+\frac{1}{\mu-1}3^d\right)}.
  \label{eq:lambda_rational}
\end{equation}
From our choices of $d$ and $\mu$ we have $3^{\nicefrac{d}{2}}-\bar{\Phi}_{n-1} > \varepsilon(\mu)$ and $3^{\nicefrac{d}{2}} > \bar{\Phi}_{n-1} $, and thus the left hand side of \eqref{eq:lambda_rational} is strictly smaller than the right hand side.
Since the denominator and the numerator of both bounds are of size polynomial in $n$, the denominator and the numerator of $\lambda$ can then be chosen to be polynomial in $n$.
Therefore, suitable $d, \mu$ and $\lambda$ can all be found in time polynomial in $n$.

Hence the game $\mathcal{G}^d_\mu$ is of polynomial size with respect to $n$ and for any $\alpha < \bar{\Phi}_{n-1} < 3^{\nicefrac{d}{2}} - \varepsilon(\mu)$ there is a unique $\alpha$-PNE where the players emulate the computation of the circuit $\bar{C}$.

In the nonexistence game $\mathcal{G}_n$ of \Cref{th:nonexistence_general_costs}, we replace $\xi$ by $\bar{\Phi}_{n-1}$ in the cost functions, making them rational-valued.
Since $\bar{\Phi}_{n-1} < \left(1+\frac{1}{\bar{\Phi}_{n-1}}\right)^{n-1}$, the $\alpha$-improving dynamics are not changed by this replacement for any $\alpha<\bar{\Phi}_{n-1}$, so that the game $\mathcal{G}_n$ does not have $\alpha$-PNE.

Finally, we replace $\Phi^2_{n-1}$ by $\bar{\Phi}_{n-1}\left(1+\frac{1}{\bar{\Phi}_{n-1}}\right)^{n-1}$ in the cost function of the dummy resource.
Since $3^{d/2}> \left(1+\frac{1}{\bar{\Phi}_{n-1}}\right)^{n-1}$, we have $\bar{\Phi}_{n-1}\left(1+\frac{1}{\bar{\Phi}_{n-1}}\right)^{n-1} < 3^d < \frac{\mu}{\mu - 1} 3^d$, thus maintaining the behavior of the output player of the circuit game.
For the players in $\mathcal{G}_n$, observe that the cost of any strategy is upper bounded by $\left(1+\frac{1}{\bar{\Phi}_{n-1}}\right)^{n-1}$, and hence the dummy resource is $\bar{\Phi}_{n-1}$-dominated.
With the above modifications, all cost functions appearing in the construction are rational-valued, while maintaining the desired $\alpha$-dynamics.

To show that $\tilde{\alpha}(n+s) < \bar{\Phi}_{n-1}$, we use the same chain of inequalities as in the proof of \Cref{th:hardness_existence_general_alt} replacing $\Phi_{n-1}$ by $\bar{\Phi}_{n-1}$, since the lower bound on $\bar{\Phi}_{n-1}$ and the upper bound on $\Phi_{n+s-1}$ still hold (see \eqref{eq:bounds-barphi}).

\end{document}